\documentclass[journal]{IEEEtranTIE}

\usepackage{graphicx}
\usepackage{cite}
\usepackage{picinpar}
\usepackage{amsmath}
\usepackage{url}
\usepackage{flushend}
\usepackage[latin1]{inputenc}
\usepackage{colortbl}
\usepackage{soul}
\usepackage{multirow}
\usepackage{pifont}
\usepackage{color}
\usepackage{alltt}
\usepackage[hidelinks]{hyperref}
\usepackage{enumerate}
\usepackage{siunitx}
\usepackage{breakurl}
\usepackage{epstopdf}
\usepackage{pbox}
\usepackage{algorithm}
\usepackage{algorithmic}

\ifCLASSOPTIONcompsoc
  \usepackage[caption=false,font=normalsize,labelfont=sf,textfont=sf]{subfig}
\else
  \usepackage[caption=false,font=footnotesize]{subfig}
\fi

\usepackage{amsmath,amssymb,amsfonts}
\usepackage{cite}
\usepackage{threeparttable}
\usepackage{algorithm}
\usepackage{algorithmic}
\usepackage{graphicx}
\usepackage{lipsum}
\usepackage{amsthm}

\begin{document}
\newtheorem{definition}{Definition}
\newtheorem{lemma}{Lemma}
\newtheorem{theorem}{Theorem}
\newtheorem{assumption}{Assumption}
\newtheorem{remark}{Remark}

\title{Recurrent Model Predictive Control: Learning an Explicit Recurrent Controller for Nonlinear Systems}

\author{
	\vskip 1em
	{Zhengyu Liu, Jingliang Duan, Wenxuan Wang, Shengbo Eben Li*, Yuming Yin, Ziyu Lin and~Bo Cheng
	}

	\thanks{This study is supported by NSF China with U20A20334 and 52072213, and it is also partially supported by Tsinghua-Toyota Joint Research Institute Cross-discipline Program. Z. Liu and J. Duan have equally contributed to this study. All correspondences should be sent to S. Li with email: lisb04@gmail.com. 
		}
	\thanks{Z. Liu, J. Duan, W. Wang,  S. Li, Z. Lin and B. Cheng are with the School of Vehicle and Mobility, Tsinghua University, Beijing, 100084, China.  {\tt\small Email: (liuzheng17, djl15, wang-wx18, linzy17)@mails.tsinghua.edu.cn; (lishbo, chengbo)@tsinghua.edu.cn}.
}
\thanks{Y. Yin is with the School of Mechanical Engineering, Zhejiang University of Technology, Zhejiang, China. {\tt\small Email: yinyuming89@gmail.com}.
}
}

\maketitle
\thispagestyle{plain}
\pagestyle{plain}
	
\begin{abstract}
This paper proposes an offline control algorithm, called Recurrent Model Predictive Control (RMPC), to solve large-scale nonlinear finite-horizon optimal control problems. It can be regarded as an explicit solver of traditional Model Predictive Control (MPC) algorithms, which can adaptively select appropriate model prediction horizon according to current computing resources, so as to improve the policy performance. Our algorithm employs a recurrent function to approximate the optimal policy, which maps the system states and reference values directly to the control inputs. The output of the learned policy network after $N$ recurrent cycles corresponds to the nearly optimal solution of $N$-step MPC. A policy optimization objective is designed by decomposing the MPC cost function according to the Bellman's principle of optimality. The optimal recurrent policy can be obtained by directly minimizing the designed objective function, which is applicable for general nonlinear and non input-affine systems. Both simulation-based and real-robot path-tracking tasks are utilized to demonstrate the effectiveness of the proposed method.

\end{abstract}



\definecolor{limegreen}{rgb}{0.2, 0.8, 0.2}
\definecolor{forestgreen}{rgb}{0.13, 0.55, 0.13}
\definecolor{greenhtml}{rgb}{0.0, 0.5, 0.0}

\section{Introduction}
\IEEEPARstart{M}{odel} Predictive Control (MPC) is a well-known method to solve finite-horizon optimal control problems online, which has been extensively investigated in various fields \cite{qin2003survey,vazquez2014model,li2014fast}. 
However, existing MPC algorithms still suffer from a major challenge: relatively low computation efficiency \cite{lee2011model}.

One famous approach to tackle this issue is the moving blocking technique, which assumes constant control input in a fixed portion of the prediction horizon. It increases the computation efficiency by reducing the number of variables to be optimized \cite{cagienard2007move}. This solution cannot guarantee control performance, system stability, and constraint satisfaction. 
In addition, Wang and Boyd (2009) proposed an early termination interior-point method to reduce the calculation time by limiting the maximum number of iterations per time step \cite{wang2009fast}. However, these online methods are still unable to meet the online computing requirement for nonlinear and large-scale systems. 

Some control algorithms choose to calculate a near-optimal explicit policy offline, and then implement it online. Bemporad \emph{et al}. (2002) first proposed the explicit MPC method to increase the computation efficiency, which partitioned the constrained state space into several regions and calculated explicit feedback control laws for each region \cite{bemporad2002explicit}. During online implementation, the onboard computer only needs to choose the corresponding state feedback control law according to the current system state, thereby reducing the burden of online calculation to some extent. Such algorithms are only suitable for small-scale systems, since the required storage capacity grows exponentially with the state dimension \cite{kouvaritakis2002needs}. 

Furthermore, significant efforts have been devoted to approximation MPC algorithms, which can reduce polyhedral state regions and simplify explicit control laws. Geyer \emph{et al}. (2008) provided an optimal merging approach to reduce partitions via merging regions with the same control law \cite{geyer2008optimal}. Jones \emph{et al}. (2010) proposed a polytopic approximation method using double description and barycentric functions to estimate the optimal policy, which greatly reduced the partitions and could be applied to any convex problem \cite{jones2010polytopic}. 
Wen \emph{et al}. (2009) proposed a piecewise continuous grid function to represent an explicit MPC solution, which reduced the requirements of storage capacity and improved online computation efficiency \cite{wen2009analytical}. Borrelli \emph{et al}. (2010) proposed an explicit MPC algorithm which can be executed partially online and partially offline\cite{borrelli2010computation}. In addition, some MPC studies employed a parameterized function to approximate the MPC controller. They updated the function parameters by minimizing the MPC cost function with a fixed prediction horizon through supervised learning or reinforcement learning \cite{aakesson2005neural,aakesson2006neural,cheng2015neural}.

Note that the policy performance and the computation time for each step usually increase with the length of the prediction horizon. The above-stated algorithms usually have to make a trade-off between control performance and computation time, and then select a conservative fixed prediction horizon to meet the requirement of real-time decision-making. However, on-board computing resources are usually dynamically changing, so these algorithms usually lead to calculation timeouts or resources waste. In other words, these algorithms cannot adapt to the dynamic allocation of computing resources and make full use of the available computing time to select the longest model prediction horizon.

In this paper, we propose an offline MPC algorithm, called Recurrent MPC (RMPC), for finite-horizon optimal control problems with large-scale nonlinearities and nonaffine inputs. Our main contributions can be summarized as below:
\begin{enumerate}
	\item A recurrent function is employed to approximate the optimal policy, which maps the system states and reference values directly to the control inputs. Compared to previous algorithms employing non-recurrent functions (such as fully connected neural networks), which are only suitable for fixed-horizon predictive control \cite{aakesson2005neural,aakesson2006neural,cheng2015neural}, the inclusion of the recurrent structure allows the algorithm to select an appropriate prediction horizon according to current computing resources. In particular, the output of the learned policy function after $N$ recurrent cycles corresponds to the nearly optimal solution of $N$-step MPC.
	\item A policy optimization objective is designed by decomposing the MPC cost function according to the Bellman's principle of optimality. The optimal recurrent policy can be obtained by directly minimizing the designed objective function. Therefore, unlike most explicit MPC algorithms \cite{bemporad2002explicit, kouvaritakis2002needs,geyer2008optimal,jones2010polytopic,wen2009analytical,borrelli2010computation} that can only handle linear systems, the proposed algorithm is applicable for nonlinear and non input-affine systems.  
	\item RMPC calculates a near-optimal recurrent policy offline, and then implement online. Experiment shows that it is over 5 times faster than the traditional online MPC algorithms \cite{Andreas2006Biegler,bonami2008algorithmic} under the same problem scale.
\end{enumerate}

The paper is organized as follows. In Section \ref{sec:pre}, we provide the formulation of the MPC problem. Section \ref{sec:RMPC} presents RMPC algorithm and proves its convergence. In Section \ref{sec:simulation}, we perform a hardware-in-the-loop (HIL) simulation to demonstrate the generalizability and effectiveness of RMPC. Section \ref{sec:experiment} verifies the performance of RMPC in a four-wheeled robot, and Section \ref{sec:conclusion} concludes this paper.

\section{Preliminaries}
\label{sec:pre}
Consider the general time-invariant discrete-time dynamic system
\begin{equation}
\label{eq.system}
    x_{i+1}=f(x_{i}, u_{i})
\end{equation}
with state $x_{i}\in \mathcal{X} \subset \mathbb{R}^{n}$, control input $u_{i}\in \mathcal{U} \subset \mathbb{R}^{m}$, 
and the system dynamics function $f:\mathbb{R}^{n} \times \mathbb{R}^{m} \to \mathbb{R}^{n}$. We assume that $f(x_i,u_i)$ is Lipschitz continuous on a compact set $\mathcal{X}$, and the system is stabilizable on $\mathcal{X}$.

The $N$-step Model Predictive Control (MPC) problem without state constraints is given as
\begin{equation}
\label{eq.valuedefinition}
\begin{aligned}
&\min_{u^{N}_{0},\cdots, u^{N}_{N-1}}  V(x_{0},r_{1:N},N)={\sum_{i=1}^{N}l(x_{i},r_{i},u^{N}_{i-1}(x_{0},r_{1:N}))}\\
&\qquad \text{s.\;t.}  \qquad  \eqref{eq.system}, \  u\in \mathcal{U},
\end{aligned}
\end{equation}
where $V(x_{0},r_{1:N},N)$ is the cost function, $x_{0}$ is initial state, $N$ is length of prediction horizon, $r_{1:N}=[r_{1},r_{2},\cdots,r_{N}]$ is reference trajectory, $u^{N}_{i-1}$ is the control input of the $i$th step in $N$-step prediction, and $l\geq0$ is the utility function. The purpose of MPC is to find the optimal control sequence to minimize the objective $V(x_{0},r_{1:N},N)$, which can be denoted as
\begin{equation}
\label{eq.control_sequence}
\begin{aligned}
	\left[{u^{N}_{0}}^*(x_{0},r_{1:N}),{u^{N}_{1}}^*(x_{0},r_{1:N}),\cdots,{u^{N}_{N-1}}^*(x_{0},r_{1:N})\right]
    \\=\mathop{\arg\min}_{u^{N}_{0},u^{N}_{1},\cdots,u^{N}_{N-1}}V(x_{0},r_{1:N},N),
\end{aligned}
\end{equation}
where the superscript $^*$ represents optimum.

\section{Recurrent Model Predictive Control}
\label{sec:RMPC}
\subsection{Recurrent Policy Function}
In practical applications, we only need to execute the first control input ${u^{N}_{0}}^*(x_{0},r_{1:N})$ of the optimal sequence in \eqref{eq.control_sequence} at each time step. Given a control problem, assume that $N_{\text{max}}$ is the maximum feasible prediction horizon. We aim to make full use of computation resources and adaptively select the longest prediction horizon $k\in[1,N_{\text{max}}]$, which means we need to calculate and store the optimal control input ${u^{k}_0}^*(x,r_{1:k})$ of $\forall x\in \mathcal{X}$, $\forall r_{1:k}$, and $\forall k\in[1,N_{\text{max}}]$ in advance. This requires us to find an efficient way to represent the policy for different prediction horizon $N\in[1,N_{\text{max}}]$ and solve it offline.

We first introduce a recurrent function, denoted as $\pi^{c} (x_{0},r_{1:c};\theta)$, to approximate the control input ${u^{c}_{0}}^*(x_{0},r_{1:c})$, where $\theta$ is the vector of function parameters and $c$ is the number of recurrent cycles of the policy function. The goal of the proposed Recurrent MPC (RMPC) algorithm is to find the optimal parameters $\theta^*$, such that
\begin{equation}
\label{eq.equation_optimal}
\begin{aligned}
    \pi^{c}(x_{0},r_{1:c};\theta^*)&={u^{c}_{0}}^*(x_{0},r_{1:c}),\\
    \forall x_{0}&\in \mathcal{X}, \forall r_{1:c} , \forall c \in [1,N_{\text{max}}].
\end{aligned}
\end{equation}
The structure of the recurrent policy function is illustrated in Fig. \ref{fig_structure}. All recurrent cycles share the same parameters $\theta$, where $h_c\in\mathbb{R}^q$ is the vector of hidden states.
\begin{figure}[htb]
\captionsetup{justification =raggedright,
              singlelinecheck = false,labelsep=period, font=small}
\centering{\includegraphics[width=0.35\textwidth]{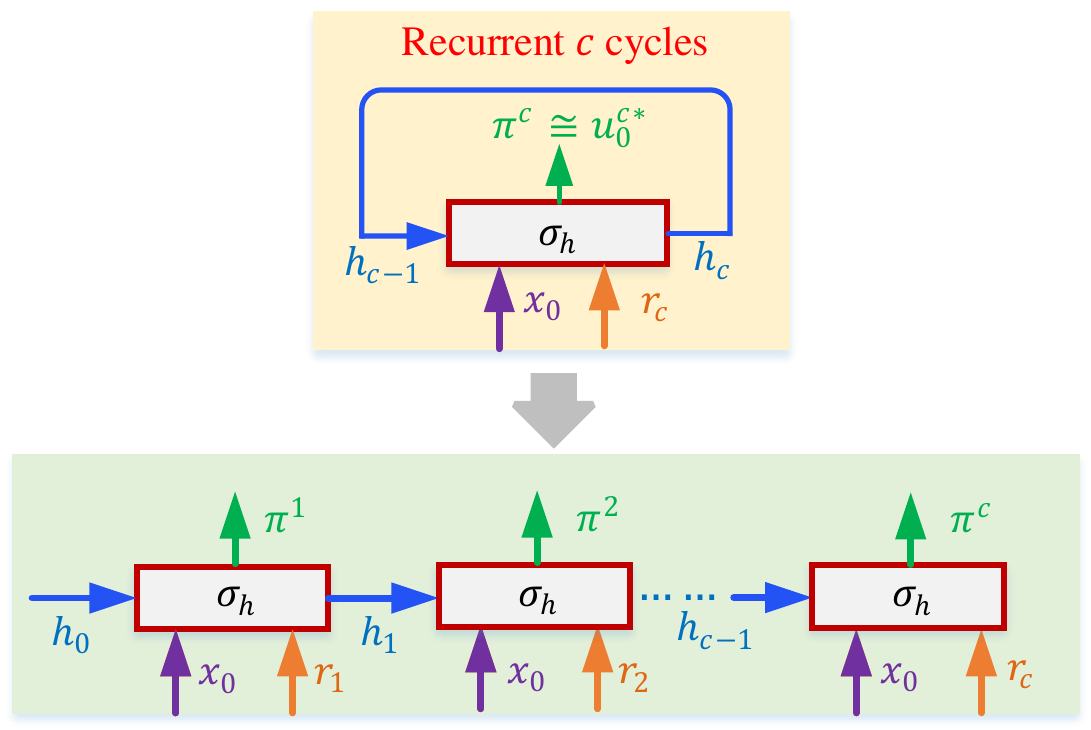}}
\caption{The structure of the recurrent policy function.}
\label{fig_structure}
\end{figure}

Each recurrent cycle is mathematically described as 
\begin{equation}
\label{eq.recurrentstructure}
\begin{aligned}
    &h_c=\sigma_h(x_{0},r_{c},h_{c-1};\theta_{h}),
    \\&\pi^c (x_{0},r_{1:c};\theta)=\sigma_y(h_c;\theta_{y}),
    \\&c\in[1,N_{\text{max}}], \theta=\theta_{h}\mathop{\cup}\theta_{y},
\end{aligned}
\end{equation}
where $h_0=0$, $\sigma_h$, and $\sigma_y$ are activation functions of hidden layer and output layer, respectively. 

As shown in Fig. \ref{fig_structure}, the recurrent policy function outputs a control input at each recurrent cycle. Assuming that we have found the optimal parameters $\theta^*$, it follows that the output of the $c$th cycle  $\pi^c(x_{0},r_{1:c};\theta^*)={u^{c}_{0}}^*(x_{0},r_{1:c})$ for $\forall c\in[1,N_{\text{max}}]$. This indicates that the more cycles, the longer the prediction horizon. In practical applications, the calculation time of each cycle $t_c$ is different due to the dynamic change of computing resource allocation (see Fig. \ref{fig_resource}). At each time step, the total time assigned to the control input calculation is assumed to be $T$. Denoting the final number of the recurrent cycles at each time step as $k$, then the corresponding control input is $\pi^k(x_{0},r_{1:k};\theta^*)$, where
\begin{equation}
\nonumber
k=\left\{
\begin{aligned}
&N_\text{max}, &\quad \sum_{c=1}^{N_{\text{max}}}t_c\le T, \\
&p, &\quad \sum_{c=1}^{p}t_c\le T < \sum_{c=1}^{p+1}t_c.
\end{aligned}
\right.
\end{equation}
Therefore, the recurrent policy is able to make full use of computing resources and adaptively select the longest prediction step $k$. In other word, the more computing resources allocated, the longer prediction horizon will be selected, which would usually lead to better control performance.
\begin{figure}[!htb]
\captionsetup{justification =raggedright,
              singlelinecheck = false,labelsep=period, font=small}
\centering{\includegraphics[width=0.4\textwidth]{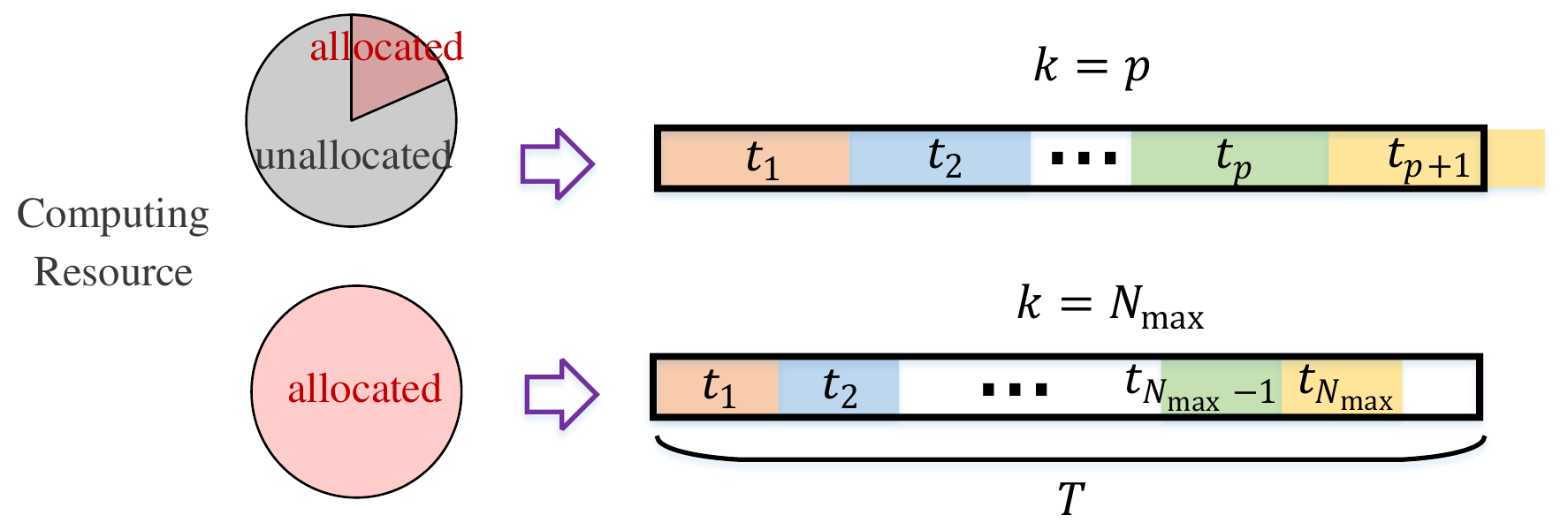}}
\caption{Maximum recurrent cycles in different cases.}
\label{fig_resource}
\end{figure}

\begin{remark}
Existing MPC algorithms usually employ non-recurrent approximation functions to represent the policies \cite{aakesson2005neural,aakesson2006neural,cheng2015neural}, which must select a fixed prediction horizon in advance. When the prediction horizon changes, the optimization problem must be reconstructed to learn a new corresponding policy. Conversely, RMPC employs recurrent function to approximate the optimal policy, which maps the system states and reference values directly to the control inputs. The use of recurrent structure allows RMPC to select an appropriate model prediction horizon according to current computing resources. The output of the learned policy network after $N$ recurrent cycles corresponds to the nearly optimal solution of $N$-step MPC.
\end{remark}

\subsection{Objective Function for Policy Learning}
To find the optimal parameters $\theta^*$ offline, we first need to represent the MPC cost function in \eqref{eq.valuedefinition} in terms of $\theta$, denoted by $V(x_{0},r_{1:N},N; \theta)$. From \eqref{eq.valuedefinition} and the Bellman's principle of optimality, the global minimum $V^*(x_{0},r_{1:N},N)$ can be expressed as
\begin{equation}
\nonumber
\begin{aligned}
   V^*(&x_{0},r_{1:N},N)
   \\&=l(x_{1},r_{1},{u^{N}_0}^*(x_0,r_{1:N}))+V^*(x_{1},r_{2:N},N-1)\\
   &=\sum_{i=1}^{2}l(x_i,r_i,{u^{N-i+1}_{0}}^*(x_{i-1},r_{i:N})) +\\
   &\qquad\qquad\qquad\qquad \qquad\qquad\qquad V^*(x_{2},r_{3:N},N-2)\\
   &\vdots\\
   &=\sum_{i=1}^{N-1}l(x_i,r_i,{u^{N-i+1}_{0}}^*(x_{i-1},r_{i:N})) +V^*(x_{N-1},r_N, 1)\\
   &=\sum_{i=1}^{N}l(x_i,r_i,{u^{N-i+1}_{0}}^*(x_{i-1},r_{i:N})).
\end{aligned}
\end{equation}
Furthermore, according to \eqref{eq.valuedefinition}, one has
\begin{equation}
\label{eq.optimal_V}
\begin{aligned}
V^*(x_{0},r_{1:N},N)&=\sum_{i=1}^{N}l(x_{i},r_{i},{u^{N}_{i-1}}^*(x_{0},r_{1:N}))\\
&=\sum_{i=1}^{N}l(x_i,r_i,{u^{N-i+1}_{0}}^*(x_{i-1},r_{i:N})).
\end{aligned}
\end{equation}
Therefore, for the same $x_0$ and $r_{1:N}$, it is clear that 
\begin{equation}
\label{eq.optimalforanyone}
{u^{N}_{i-1}}^*(x_{0},r_{1:N})={u^{N-i+1}_{0}}^*(x_{i-1},r_{i:N}),\quad \forall i\in[1,N].
\end{equation}This indicates that the $i$th optimal control input ${u^{N}_{i-1}}^*(x_{0},r_{1:N})$ in \eqref{eq.control_sequence} can be regarded as the optimal control input of the $N$-$i$+$1$-step MPC control problem with initial state $x_{i-1}$. Hence, by replacing all $u^{N}_{i-1}(x_{0},r_{1:N})$ in \eqref{eq.valuedefinition} with $u^{N-i+1}_{0}(x_{i-1},r_{i:N})$, the cost function of $N$-step MPC can be rewritten as  $$V(x_{0},r_{1:N},N)=\sum_{i=1}^{N}l(x_i,r_i,u^{N-i+1}_{0}(x_{i-1},r_{i:N})).$$
Immediately, we can obtain the $N$-step cost function in terms of $\theta$:
\begin{equation}
\label{eq.appro_V}
V(x_{0},r_{1:N},N;\theta)=\sum_{i=1}^{N}l(x_i,r_i,\pi^{N-i+1}(x_{i-1},r_{i:N};\theta))
\end{equation}
Fig. \ref{fig:objective} illustrates the reshaped $N$-step cost function intuitively.
\begin{figure}[!htb]
\captionsetup{justification =raggedright,
              singlelinecheck = false,labelsep=period, font=small}
\centering{\includegraphics[width=0.4\textwidth]{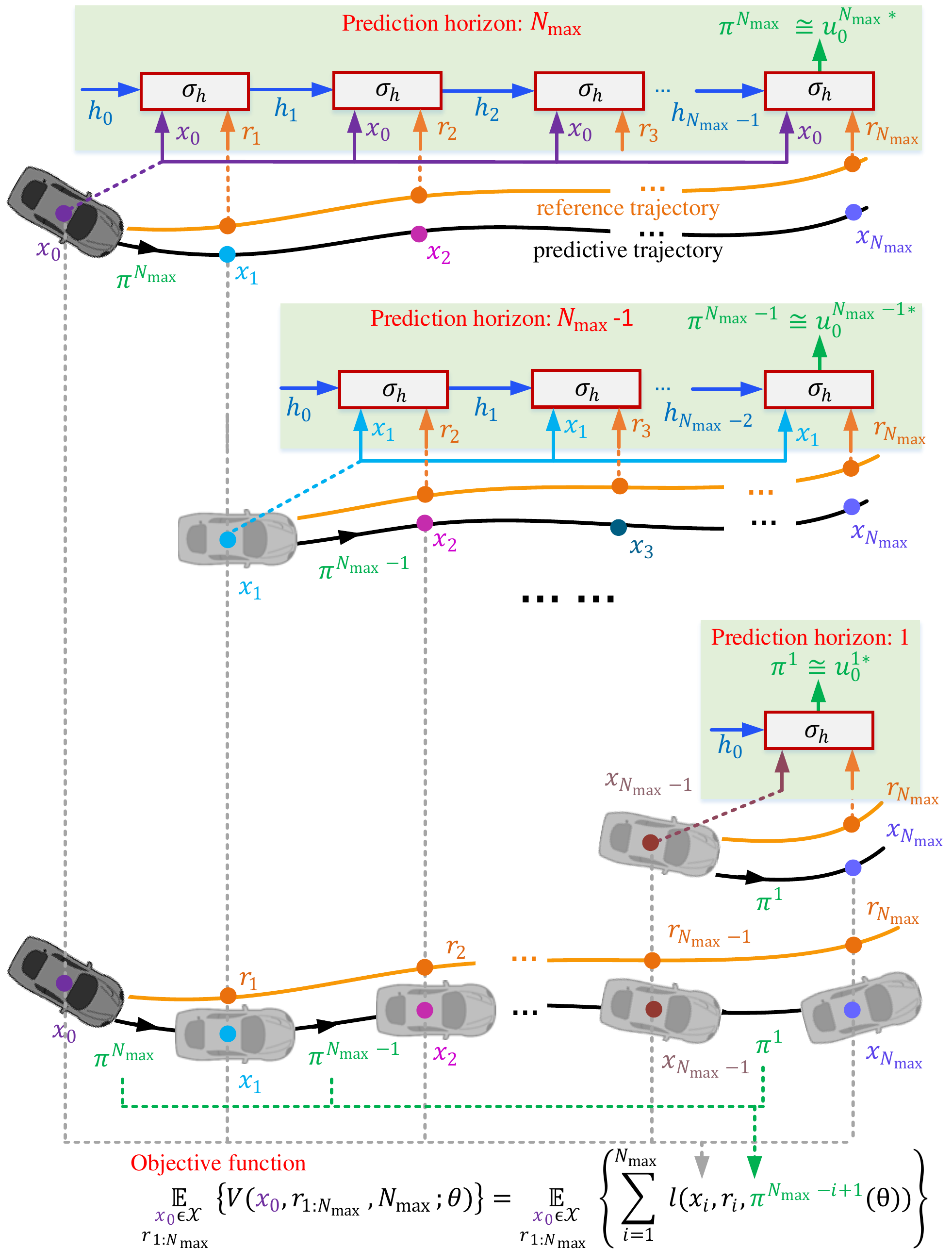}}
\caption{Illustration of the objective function of RMPC.}
\label{fig:objective}
\end{figure}

To find the optimal parameters $\theta^*$ that make \eqref{eq.equation_optimal} hold, we can construct the following objective function:
\begin{equation}
\label{eq.lossfunction}
    J(\theta)=\mathop{\mathbb{E}}_{\substack{x_0\in\mathcal{X}, r_{1:{N_{\text{max}}}}}}\Big\{V(x_{0},r_{1:{N_{\text{max}}}},N_{\text{max}};\theta)\Big\}.
\end{equation}
Therefore, we can update $\theta$ by directly minimizing $J(\theta)$. The policy update gradients can be derived as
\begin{equation}   
\label{eq.updatagradient}
\begin{aligned}
\frac{\text{d}J}{\text{d}\theta}&=\mathop{\mathbb{E}}_{
\small\begin{array}{ccc}
\small x_0\in\mathcal{X}, r_{1:N_{\text{max}}}\\
\end{array} } \Big\{\frac{\text{d}V(x_{0},r_{1:N_{\text{max}}},N_{\text{max}};\theta)}{\text{d}\theta}\Big\},\\
\end{aligned}
\end{equation}
where
\begin{equation}
\nonumber
\begin{aligned}
&\frac{\text{d}V(x_{0},r_{1:N},N_{\text{max}};\theta)}{\text{d}\theta}=\\
&\qquad\qquad\qquad\sum_{i=1}^{N_{\text{max}}}\frac{\text{d}l(x_{i},r_{i},\pi^{N_{\text{max}}-i+1}(x_{i-1},r_{i:N_{\text{max}}};\theta))}{\text{d} \theta}.
\end{aligned}
\end{equation}
Denoting $\pi^{N_{\text{max}}-i+1}(x_{i-1},r_{i:N_{\text{max}}};\theta)$ as $\pi^{N_{\text{max}}-i+1}$, $l(x_{i},r_{i},\pi^{N_{\text{max}}-i+1}(x_{i-1},r_{i:N_{\text{max}}};\theta))$ as $l_{i}$, $\frac{\mathrm{d}x_i}{\mathrm{d}\theta}$ as $\phi_{i}$ and $\frac{\mathrm{d}\pi^{N_{\text{max}}-i+1}}{\mathrm{d}\theta}$ as $\psi_{i}$, we further have
\begin{equation}
\nonumber
\frac{\text{d} V(x_{0},r_{1:N_{\text{max}}},N_{\text{max}};\theta)}{\text{d} \theta}=\sum_{i=1}^{N_{\text{max}}}
\Big\{\frac{\partial l_{i}}{\partial x_{i}}\phi_{i}+\frac{\partial l_{i}}{\partial \pi^{N_{\text{max}}-i+1}}\psi_{i}\Big\},
\end{equation}
where
\begin{equation}
\nonumber
\begin{aligned}
\phi_{i} &=\frac{\partial f(x_{i-1},\pi^{N_{\text{max}}-i+1})}{\partial x_{i-1}}\phi_{i-1} +\frac{\partial f(x_{i-1},\pi^{N_{\text{max}}-i+1})}{\partial \pi^{N_{\text{max}}-i+1}}\psi_{i},
\end{aligned}
\end{equation} 
with $\phi_0=0$, and
\begin{equation}
\nonumber
\psi_{i} =\frac{\partial\pi^{N_{\text{max}}-i+1}}{\partial x_{i-1}}\phi_{i-1} + \frac{\partial\pi^{N_{\text{max}}-i+1}}{\partial \theta}.
\end{equation} 
Fig. \ref{fig_gradient} visually shows the backpropagation path of the policy gradients.

\begin{figure}[!htb]
\captionsetup{justification =raggedright,
              singlelinecheck = false,labelsep=period, font=small}
\centering{\includegraphics[width=0.4\textwidth]{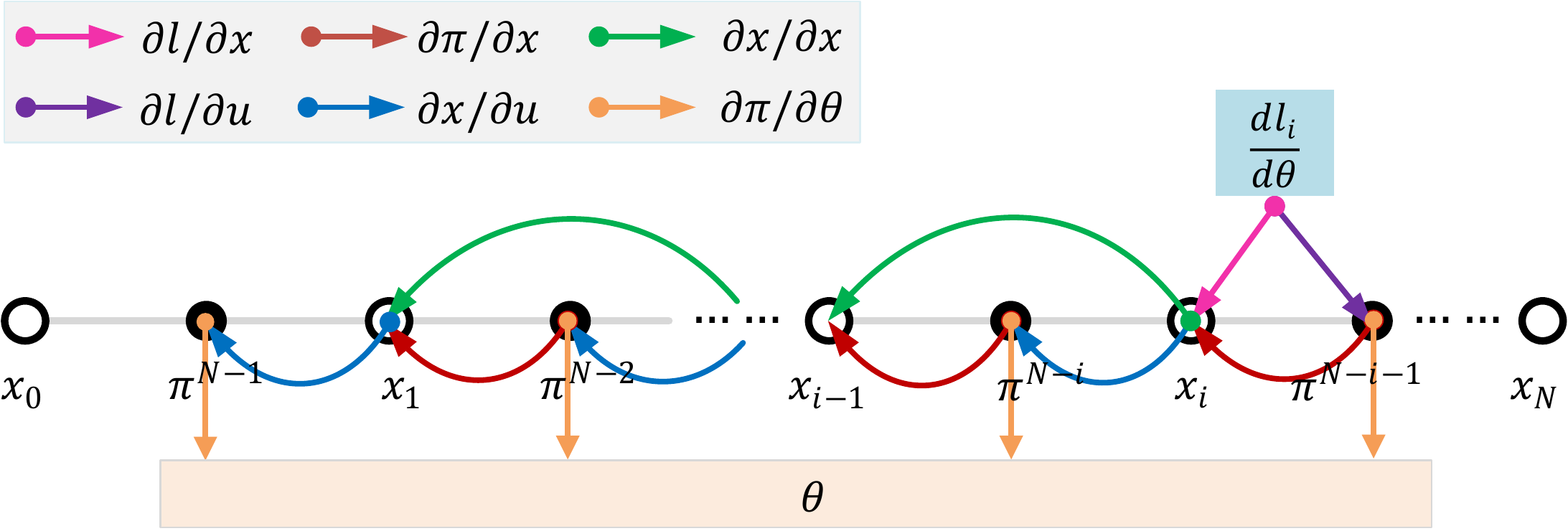}}
\caption{Backprogagation Path of the Policy Update Gradients}
\label{fig_gradient}
\end{figure}

Taking the Gradient Descent (GD) method as an example, the policy update rule is
\begin{equation}
\label{eq.update_rule}
\begin{aligned}
\theta_{K+1} &= -\alpha_{\theta} \frac{\text{d}J}{\text{d}\theta} + \theta_K,
\end{aligned}
\end{equation}
where $\alpha_{\theta} $ denotes the learning rate and $K$ indicates $K$th iteration. 

The pseudo-code and diagram of the proposed RMPC algorithm are shown in  Algorithm \ref{alg:RMPC} and Fig. \ref{fig:flowchart}. 

\begin{algorithm}[!htb]
\caption{RMPC algorithm}
\label{alg:RMPC}
\begin{algorithmic}
   \STATE Given an appropriate learning rate $\alpha_\theta$ and an arbitrarily small positive number $\epsilon$.
   \STATE Initial with arbitrary $\theta_0$
\REPEAT 
\STATE Randomly select $x_0\in \mathcal{X}$ and the corresponding  $r_{1:N_{\text{max}}}$
\STATE Calculate $\frac{\text{d}J(\theta_K)}{\text{d}\theta_K}$ using \eqref{eq.updatagradient}
\STATE Update policy function using \eqref{eq.update_rule}
\UNTIL $|J(\theta_{K+1})-J(\theta_{K})|\le \epsilon$
\end{algorithmic}
\end{algorithm}

\begin{figure}[!htb]
\captionsetup{justification =raggedright,
              singlelinecheck = false,labelsep=period, font=small}
\centering{\includegraphics[width=0.3\textwidth]{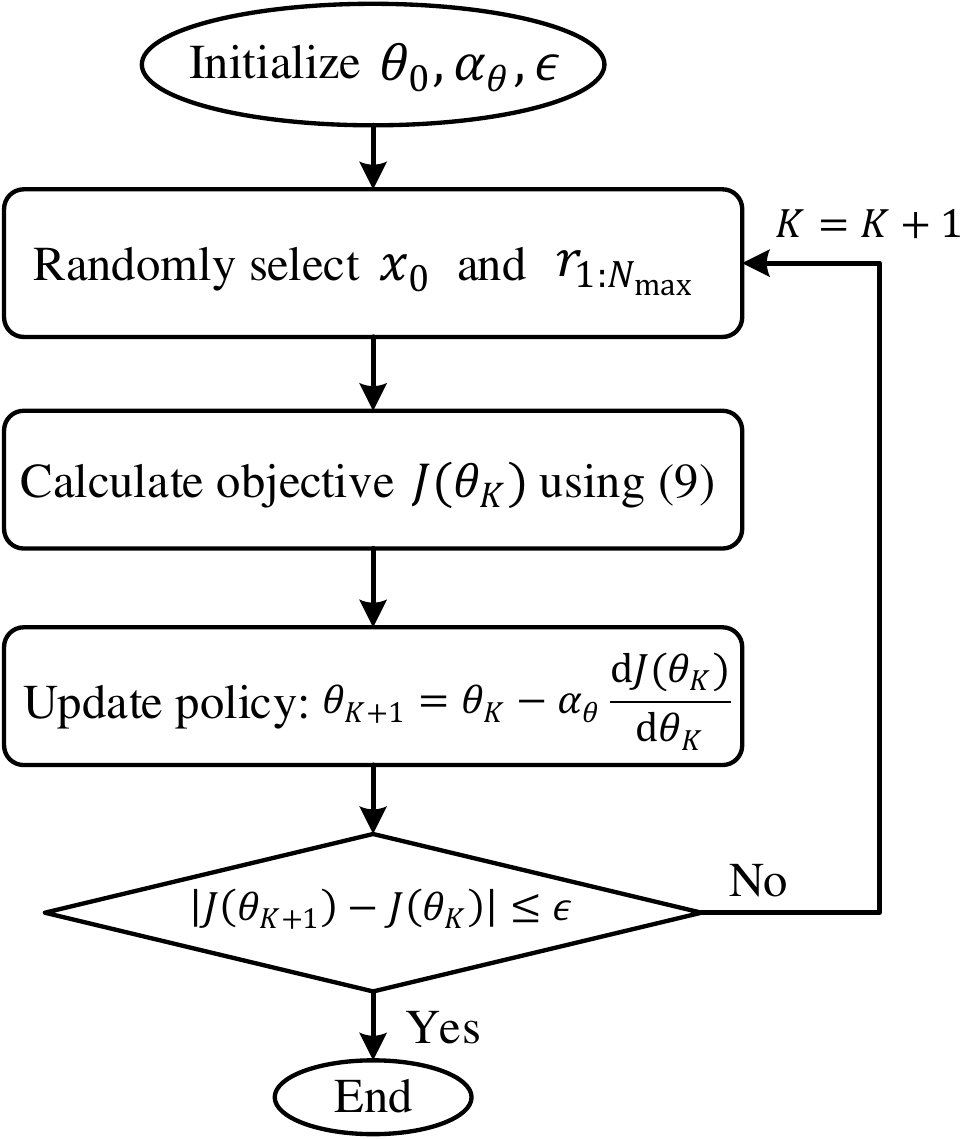}}
\caption{The training flowchart of RMPC.}
\label{fig:flowchart}
\end{figure}

\begin{remark}
Most existing explicit MPC algorithms \cite{bemporad2002explicit, kouvaritakis2002needs,geyer2008optimal,jones2010polytopic,wen2009analytical,borrelli2010computation} can only handle linear systems. As a comparison, RMPC is applicable for general nonlinear and non input-affine systems, because the nearly optimal recurrent policy can be obtained by directly minimizing the designed objective function using policy gradient methods.
\end{remark}

\begin{remark}
The objective function of RMPC in \eqref{eq.lossfunction} is identical to the traditional MPC problem \eqref{eq.valuedefinition}. The only difference between them is that RMPC aims to find an explicit nearly optimal recurrent policy rather than numerical solutions. Therefore, RMPC can be regarded as a special explicit solver for the traditional MPC problem in \eqref{eq.valuedefinition}.
\end{remark}

\subsection{Convergence and Optimality}
There are many types of recurrent functions belonging to the structure defined in \eqref{eq.recurrentstructure}, and the recurrent neural network (RNN) is the most commonly used one. In recent years, deep RNNs have been successfully implemented in many fields, such as natural language processing and system control, attributing to their ability to process sequential data \cite{mikolov2010recurrent,li2017novel}. Next, we will show that as the iteration index $K\rightarrow\infty$, the optimal policy $\pi^c (x_{0},r_{1:c};\theta^*)$ that makes \eqref{eq.equation_optimal} hold can be achieved using Algorithm \ref{alg:RMPC}, as long as $\pi^c (x_{0},r_{1:c};\theta)$ is an over-parameterized RNN. The over-parameterization means that the number of hidden neurons and layers is sufficiently large. Before the main theorem, the following lemma and assumption need to be introduced. 

\begin{lemma}[Universal Approximation Theorem\cite{li1992approximation,schafer2007recurrent,hammer2000approximation}]
\label{lemma.ability}
Consider a sequence of finite functions $\{F^i(y^i)\}_{i=1}^n$, where $n$ is the number of functions, $y^i=[y_1,y_2,\hdots,y_i]\in\mathbb{R}^i$, $i\in [1,n]$ is the input dimension, and $F^i(y^i): \mathbb{R}^i\rightarrow \mathbb{R}^d$ is a continuous function on a compact set. Describe an RNN ${G}^c (y^c;W,b)$ as
\begin{equation}
\nonumber
\begin{aligned}
    &h_c=\sigma_h(W_y^\top y^{c}+U_h^\top h_{c-1}+b_y),\\
    &{G}^c (y^c;W,b)=\sigma_y(W_h^\top h_{c}+b_h),
\end{aligned}
\end{equation}
where $c$ is the number of recurrent cycles,  $W=W_h\mathop{\cup}W_y$, $b=b_h\mathop{\cup}b_y$ and $U_h$  are parameters, $\sigma_h$ and  $\sigma_y$ are activation functions.  Supposing ${G}^c (y^c;W,b)$ is over-parameterized, for any $\{F^i(y^i)\}_{i=1}^n$, $\exists U_h, W, b$, such that
\begin{equation}
\nonumber
 \left \| {G}^c(y^c; W,b)-F^c(y^c) \right \|_{\infty} \le \epsilon, \quad \forall y^c\in\mathbb{R}^c, c \in [1,n],
\end{equation} 
where $\epsilon \in\mathbb{R^+}$ is an arbitrarily small error.
\end{lemma}

The reported experimental results and theoretical proofs have shown that the straightforward optimization methods, such as GD and Stochastic GD (SGD), can find global minima of most training objectives in polynomial time if the approximate function is an over-parameterized NN or RNN \cite{allen2019convergence,du2019gradient}. Based on this fact, we make the following assumption.

\begin{assumption} 
\label{assumption.global}
If the approximate function is an over-parameterized RNN, the global minimum of the objective function in \eqref{eq.lossfunction} can be found using an appropriate optimization algorithm such as SGD \cite{allen2019convergencernn}.
\end{assumption}

Now, we are ready to show the convergence and optimality of RMPC.

\begin{theorem}[Recurrent Model Predictive Control]
\label{theorem.optimality}
Suppose $\pi^c (x_{0},r_{1:c};\theta)$ is an over-parameterized RNN. Through Algorithm \ref{alg:RMPC}, any initial parameters $\theta_0$ will converge to $\theta^*$, such that \eqref{eq.equation_optimal} holds.
\end{theorem}

\begin{proof}
By \eqref{eq.lossfunction}, we have
\begin{equation}
\nonumber
\begin{aligned}
\min_{\theta}J(\theta)&=\min_{\theta}\mathop{\mathbb{E}}_{\substack{x_0\in\mathcal{X},r_{1:{N_{\text{max}}}}}}\Big\{V(x_{0},r_{1:{N_{\text{max}}}},N_{\text{max}};\theta)\Big\}\\
&\ge\mathop{\mathbb{E}}_{\substack{x_0\in\mathcal{X},r_{1:{N_{\text{max}}}}}}\Big\{\min_{\theta}V(x_{0},r_{1:{N_{\text{max}}}},N_{\text{max}};\theta)\Big\}.
\end{aligned}
\end{equation}
By Lemma \ref{lemma.ability}, there always $\exists \theta^{\dagger}$, such that 
\begin{equation}
\centering
\nonumber
\theta^{\dagger}=\arg\min_{\theta}V(x_{0},r_{1:N_{\text{max}}},N_{\text{max}};\theta),\quad \forall x_0\in \mathcal{X},r_{1:N_{\text{max}}}.
\end{equation}
Then, it directly follows that
\begin{equation}
\centering
\nonumber
J(\theta^{\dagger})=\min_\theta J(\theta).
\end{equation}
Furthermore, according to \eqref{eq.optimal_V}, \eqref{eq.optimalforanyone}, and the Bellman's principle of optimality, $\theta^\dagger$ can also make \eqref{eq.equation_optimal} hold, i.e., $\theta^\dagger=\theta^*$. Note that $\theta^\dagger$ may not be unique.
From Assumption \ref{assumption.global}, we can always find $\theta^*$  by repeatedly minimizing $J(\theta)$ in \eqref{eq.lossfunction} using \eqref{eq.update_rule}, which completes the proof.
\end{proof}

Thus, we have proven that the RMPC algorithm can converge to $\theta^*$. In other words, it can find the explicit nearly optimal policy of MPC with different prediction horizons, whose output after $c$th recurrent cycles corresponds to the nearly optimal solution of $c$-step MPC.

\begin{remark}
Theorem \ref{theorem.optimality} shows that under mild assumptions, the output of the converged policy of RMPC is exactly the optimal solution of the traditional MPC problem in \eqref{eq.valuedefinition}. This means that if RMPC converges to the nearly optimal solution, it would inherit the stability property of the original MPC problem. The stability conditions for the case where the approximate error of the learned recurrent policy cannot be ignored will be established in further studies.
\end{remark}

\section{Simulation Verification}
\label{sec:simulation}

In order to evaluate the performance of the proposed RMPC algorithm, we choose the vehicle lateral control problem in the path tracking task as an example \cite{li2017driver}. It is a nonlinear and non-affine control problem and a widely-used verification and application task for MPC \cite{duan2021adaptive, cheng2020model,ji2016path}. 

\subsection{Overall Settings}
The recurrent policy network is trained offline on the PC, and then deployed to the industrial personal computer (IPC). The vehicle dynamics used for policy training are different from the controlled plant, which is provided by the Carsim simulator \cite{benekohal1988carsim}. For online applications, the IPC-controller gives the control signal to the plant according to the state information and the reference trajectory. The plant feeds back the state information to the IPC-controller, so as to realize the closed-loop control process. The feedback scheme of the HIL experiment is depicted in Fig. \ref{fig_HIL}. The type of IPC-controller is ADLINK MXE-5501, equipped with Intel i7-6820EQ CPU and 8GB RAM, which is 
used as a vehicle onboard controller \cite{Chaoyi2019System}. The plant is a real-time system, simulated by the vehicle dynamic model of CarSim.
The longitudinal speed is assumed to be constant, $v_x=16 \text{m/s}$, and the expected trajectory is shown in Fig. \ref{f:comparison_linear}. The system states and control inputs of this problem are listed in Table \ref{tab.state}, and the vehicle parameters are listed in Table \ref{tab.parameters}.

\begin{figure}[!htb]
\captionsetup{
              singlelinecheck = false,labelsep=period, font=small}
\centering{\includegraphics[width=0.45\textwidth]{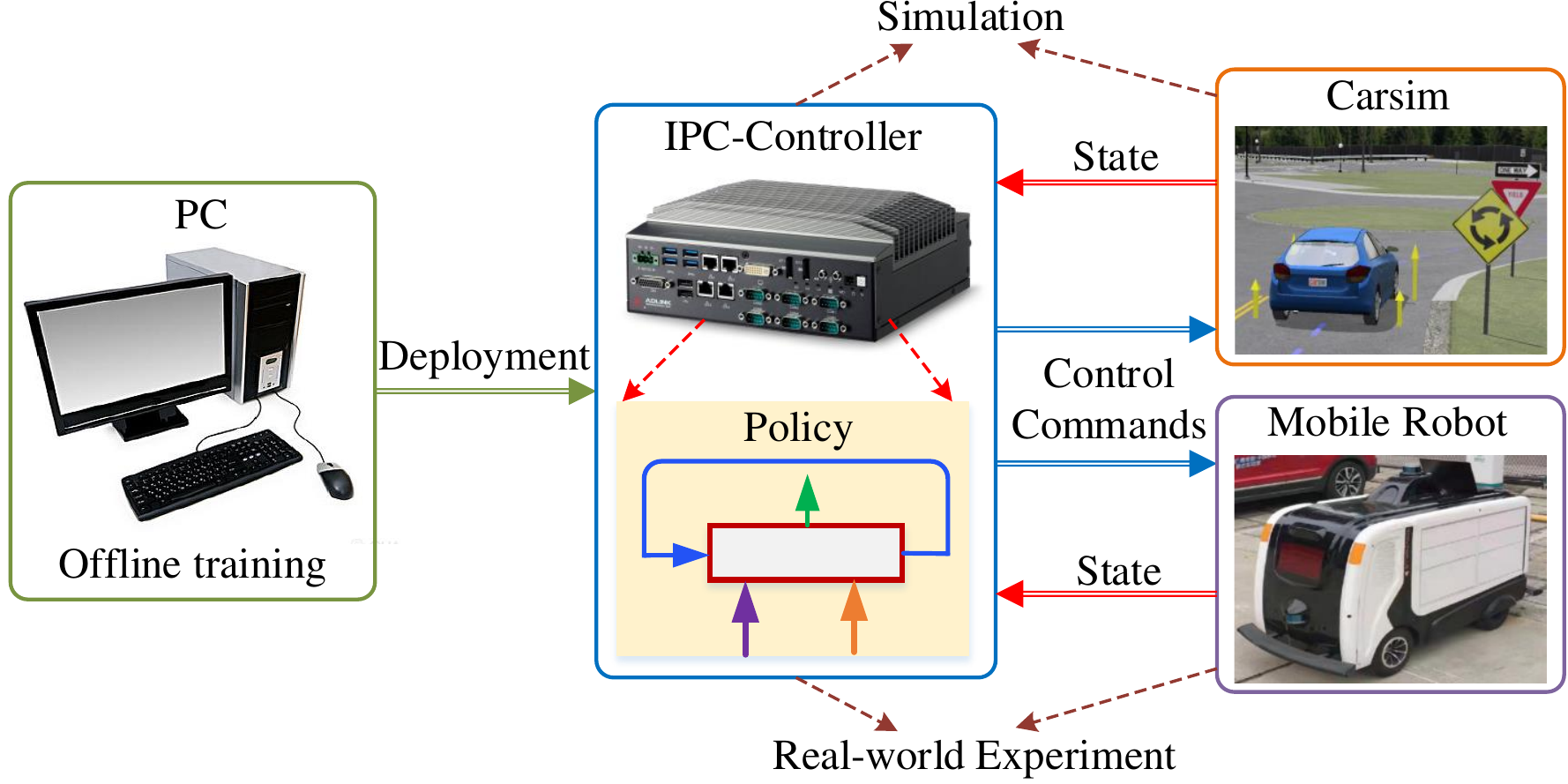}}
  \caption{Schematic view of the experimental setup.}
  \label{fig_HIL}
\end{figure}

\begin{table}[!htb]
	\caption{State and control input}
	\centering
	\label{tab.state}
		\begin{tabular}{l l l l}
			\hline\hline
			Mode &Name &Symbol&Unit\\ \hline
		state&Lateral velocity at center of gravity (CG) &$v_y$ & [m/s]  \\
		& Yaw rate  &$\omega_r$ & [rad/s] \\
		& Yaw angle &$\phi$ & [rad]\\
		& Lateral Position&$y$ & [m] \\
		input & Front wheel angle &$\delta$ & [rad] \\ 
			\hline\hline
		\end{tabular}
\end{table}

\begin{table}[!htb]
	\caption{Vehicle Parameters}
	\centering
	\label{tab.parameters}
			\begin{tabular}{l l l}
			\hline\hline
Name &Symbol&	Unit\\\hline
Longitudinal velocity at CG &$v_x$ & 16 [m/s] \\
Front tire cornering stiffness &$k_1$ & -88000 [N/rad] \\
Rear tire cornering stiffness &$k_2$ & -94000 [N/rad] \\
Mass &$m$ & 1500 [kg] \\
Distance from CG to front axle &$a$ & 1.14 [m] \\
Distance from CG to rear axle &$b$ & 1.40 [m] \\
Polar moment of inertia at CG &$I_z$ & 2420 [kg$\cdot\mathrm{m}^2$] \\
Tire-road friction coefficient &$\mu$ & 1.0 \\
System frequency &$f$ & 20 [Hz] \\ 
			\hline\hline
		\end{tabular}
\end{table}


\subsection{Problem Description}
The offline policy is trained based on the nonliner and non input-affine vehicle dynamics:
\begin{equation}
\nonumber
x = 
\begin{bmatrix}
  y  \\
  \phi \\
  v_y \\
  \omega_r
\end{bmatrix}, u =\delta, x_{i+1}=
\begin{bmatrix}
v_x \sin\phi + v_y \cos\phi\\
\omega_r\\
  \frac{F_{yf}\cos\delta + F_{yr}}{m} - v_x \omega_r\\
  \frac{aF_{yf}\cos\delta - bF_{yr}}{I_{z}}
   
\end{bmatrix}\frac{1}{f} + x_i,
\end{equation} 
where $F_{yf}$ and $F_{yr}$ are the lateral tire forces of the front and rear tires, respectively \cite{kong2015kinematic}. The lateral tire forces can be approximated according to the Fiala tire model
\begin{equation}
\nonumber
F_{y\#}  = 
\left\{
\begin{aligned}
&-C_\#\tan\alpha_\#\Big(\frac{C_\#^2(\tan\alpha_\#)^2}{27(\mu_\# F_{z\#})^2} -  \frac{{C_\#}\left |\tan\alpha_\# \right |}{3\mu_\# F_{z\#}} + 1\Big),\\ &\qquad \qquad \qquad \qquad \qquad \qquad |\alpha_\#|\le|\alpha_{\text{max},\#}|,\\
&  \mu_\# F_{z\#},&\\
&\qquad \qquad \qquad \qquad \qquad \qquad |\alpha_\#|>|\alpha_{\text{max},\#}|,
\end{aligned}
\right.
\end{equation} 
where $\alpha_{\#}$ is the tire slip angle, $F_{z\#}$ is the tire load, $\mu_\#$ is the friction coefficient, and the subscript $\# \in \{f,r\}$ represents the front or rear tires. The slip angles can be calculated from the relationship between the front/rear axle and the center of gravity (CG): 
\begin{equation}
\nonumber
\alpha_f = \arctan (\frac{v_y+a\omega_r}{v_x})-\delta, \quad \alpha_r = \arctan (\frac{v_y-b\omega_r}{v_x}). 
\end{equation} 
The loads on the front and rear tires can be approximated by: 
\begin{equation}
\nonumber
F_{zf} = \frac{b}{a+b}mg, \quad F_{zr} = \frac{a}{a+b}mg.
\end{equation} 
The utility function of this problem is set to be 
\begin{equation}
\nonumber
l(x_i,{r_i},u_{i-1}) = ([1,0,0,0]x_i-r_i)^2+10{u_{i-1}}^2+([0,0,0,1]{x_i})^2.
\end{equation}

Therefore, the policy optimization problem of this example can be formulated as:
\begin{equation}
\nonumber
\centering
\begin{aligned}
\min_\theta  &  \mathop{\mathbb{E}}_{
\small\begin{array}{ccc}
\small x_0\in\mathcal{X}, r_{1:{N_{\text{max}}}}\\
\end{array} } \Big\{V(x_{0},r_{1:{N_{\text{max}}}},{N_{\text{max}}};\theta)\Big\}\\
s.t. \quad &x_{i} = f(x_{i-1},u_{i-1}), \\
&u_{\text{min}}\leq u_{i-1} \leq u_{\text{max}},\\
&i\in[1,N_{\text{max}}].
\end{aligned}
\end{equation}
where $V(x_{0},r_{1:{N_{\text{max}}}},{N_{\text{max}}};\theta)=\sum_{i=1}^{N_{\text{max}}}l(x_i,r_i,u_{i-1})$, $u_{i-1}=\pi^{{N_{\text{max}}}-i+1}(x_{i-1},r_{i:{N_{\text{max}}}};\theta)$, $N_{\text{max}}=15$, $u_{\text{min}}=-0.2$ rad, and $u_{\text{max}}=0.2$ rad. 

\subsection{Algorithm Details}
The policy function is represented by a variant of RNN, called GRU (Gated Recurrent Unit). The input layer is composed of the states, followed by 4 hidden layers using rectified linear units (RELUs) as activation functions, with $128$ units per layer. The output layer is set as a $tanh$ layer, multiplied by $0.2$ to confront bounded control inputs. We use the Adam optimization method to update the policy with the learning rate of $2\times10^{-4}$ and the batch size of $256$.

\subsection{Result Analysis}
Given a nonlinear MPC problem, we can directly solve it with some optimization solvers, such as IPOPT \cite{Andreas2006Biegler} and BONMIN \cite{bonami2008algorithmic}, whose numerical solutions can be approximately regarded as the optimal policy. In the sequel, both IPOPT and BONMIN are implemented in a symbolic framework, called CasADi \cite{2018CasADi}. 

Since RMPC is an explicit solver of the traditional MPC problem, if the control inputs of RMPC are approximately equal to the optimal numerical solutions under different prediction horizons, we can immediately show that RMPC can adaptively choose the longest prediction step. We run Algorithm \ref{alg:RMPC} for 10 runs and calculate the policy error $e_N$ between the solution of IPOPT and RMPC at each iteration with different prediction steps ($N\in[1,15]$),
\begin{equation}
\nonumber
e_N= \mathop{\mathbb{E}}_{
\small\begin{array}{ccc}
\small x_0\in\mathcal{X}, r_{1:N}\\
\end{array} }
\left[
\frac{ \vert {u_{0}^{N}}^*(x_0,r_{1:N})-{\pi^N}(x_0,r_{1:N};\theta) \vert }{{u^N_{\text{max}}}^*-{u^N_{\text{min}}}^*}\right],
\end{equation}
where ${{u^N_{\text{max}}}^*}$ and ${{u^N_{\text{min}}}^*}$ are respectively the maximum and minimum value of ${u_{0}^{N}}^*(x_0,r_{1:N})$ for $\forall x_0 \in \mathcal{X}$, $\forall N \in [1,15]$. Fig. \ref{fig_err_train} plots policy error curves during training with different prediction steps. It is clear that all the policy errors decrease rapidly to a small value during the training process. In particular, after $10^4$ iterations, policy errors for all  $N\geq5$ 
reduce to less than 2\%. This indicates that Algorithm  \ref{alg:RMPC} has the ability to find the nearly optimal explicit policy of MPC problems with different prediction horizons.
\begin{figure}[!htb]
\captionsetup{
              singlelinecheck = false,labelsep=period, font=small}
\centering{\includegraphics[width=0.4\textwidth]{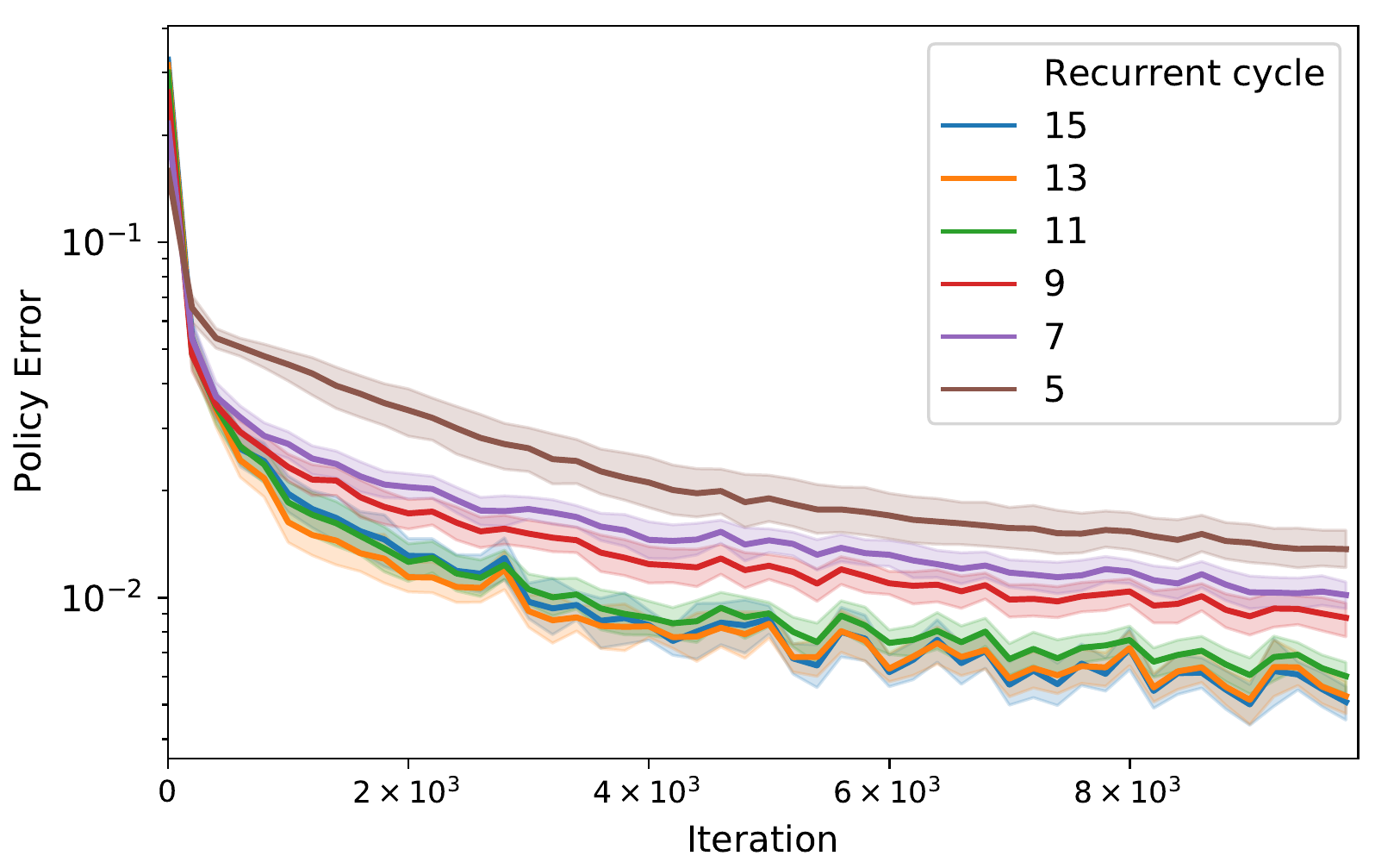}}
  \caption{Policy errors during training. Solid lines are average values over 10 runs. Shaded regions correspond to 95\% confidence interval.}
  \label{fig_err_train}
\end{figure}

Fig. \ref{fig_time} compares the calculation efficiency of RMPC and optimization solvers in online applications. It is obvious that the calculation time of the optimization solvers is much longer than RMPC, and the gap increases with the number of prediction steps. Specifically, when $N=c=15$, the IPOPT solver is about 5 times slower than RMPC (IPOPT for $26.2$ms, RMPC for $4.7$ms). This demonstrates the online effectiveness of the RMPC method.
\begin{figure}[!htb]
\captionsetup{justification =raggedright,
              singlelinecheck = false,labelsep=period, font=small}
\centering{\includegraphics[width=0.4\textwidth]{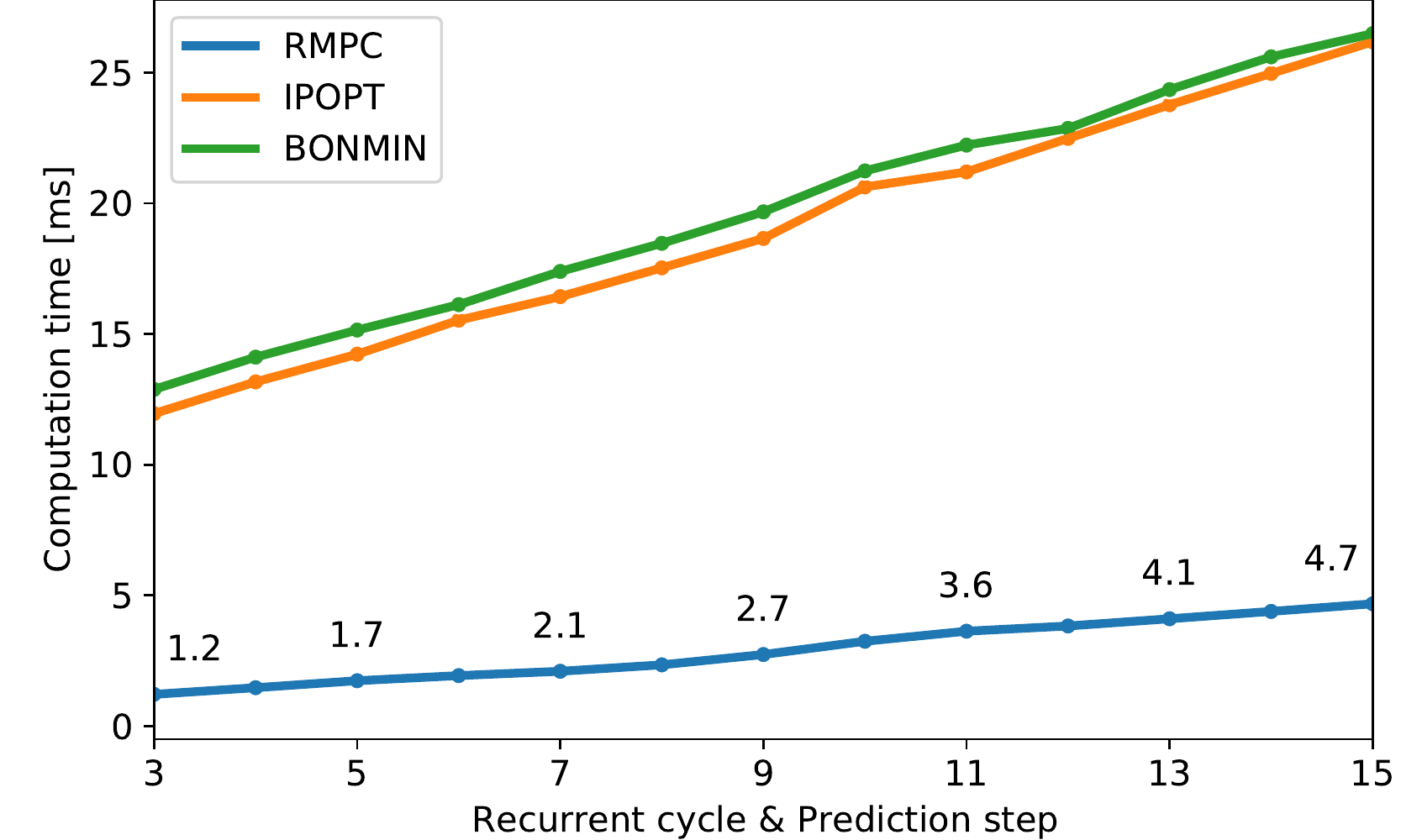}}
  \caption{RMPC vs other solvers in terms of online computation time.}
\label{fig_time}
\end{figure}

Fig. \ref{fig_loss} compares the policy performance of IPOPT and RMPC with different prediction horizons. The policy performance is measured by the cost-to-go of 200 steps (10s) during simulation starting from randomly initialized states, i.e., 
\begin{equation}
\label{eq.loss_simulation}
\centering
L=\sum_{i=1}^{200} l(x_i,r_i,u_{i-1}). 
\end{equation}
For all prediction domains $N$, RMPC performs as well as the solution of the IPOPT solver. Besides, more recurrent cycles (or long prediction steps) help to reduce the cost-to-go $L$.
\begin{figure}[!htb]
\captionsetup{
              singlelinecheck = false,labelsep=period, font=small}
\centering{\includegraphics[width=0.4\textwidth]{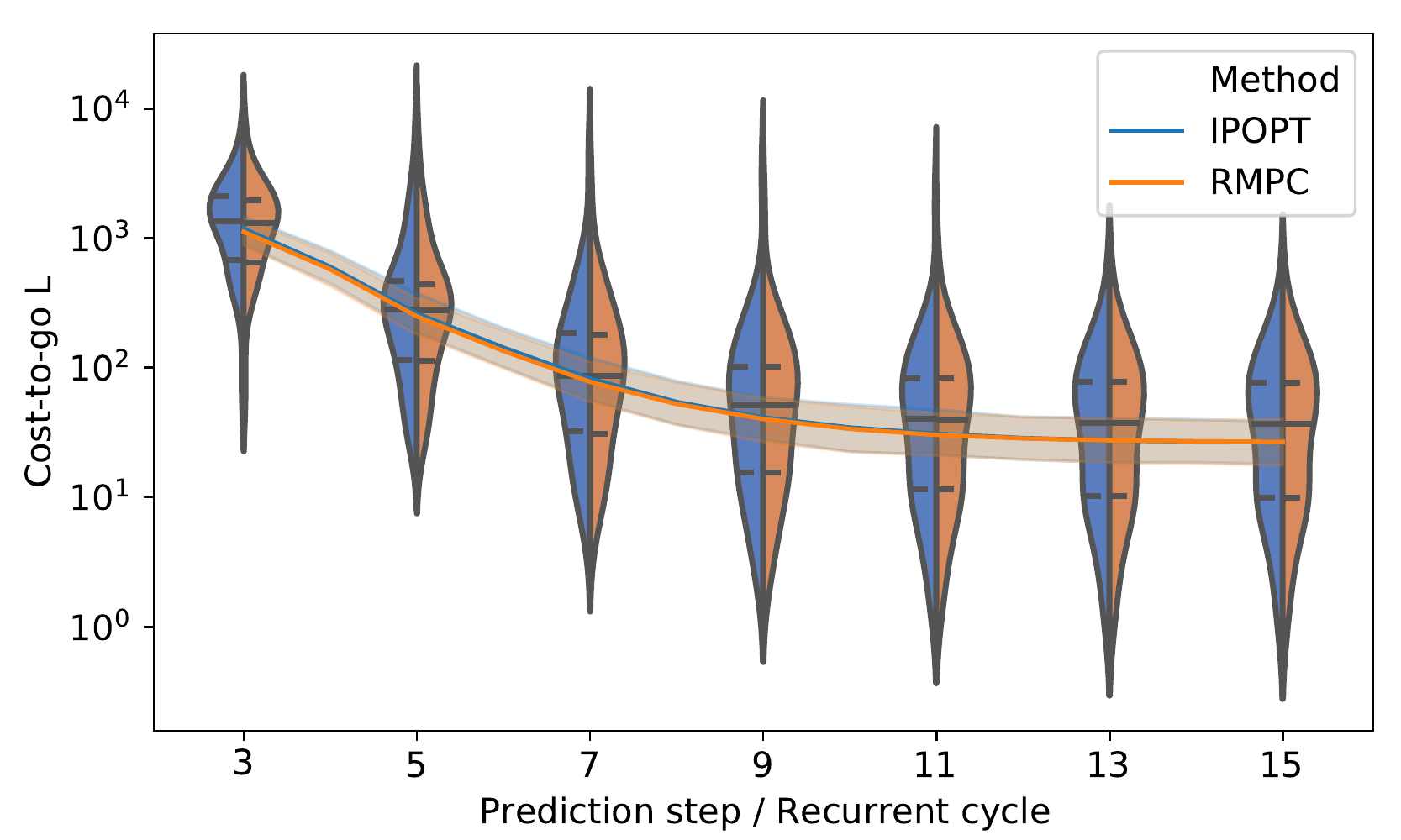}}
\caption{Performance comparison between RMPC and IPOPT. Solid lines are average values over 50 initialized states. Shaded regions correspond to 95\% confidence interval.}
\label{fig_loss}
\end{figure}

\begin{figure*}[!htb]
\centering
\captionsetup{justification =raggedright,
              singlelinecheck = false,labelsep=period, font=small}
\captionsetup[subfigure]{justification=centering}
\subfloat[\label{subFig:st1}]{\includegraphics[width=0.22\textwidth]{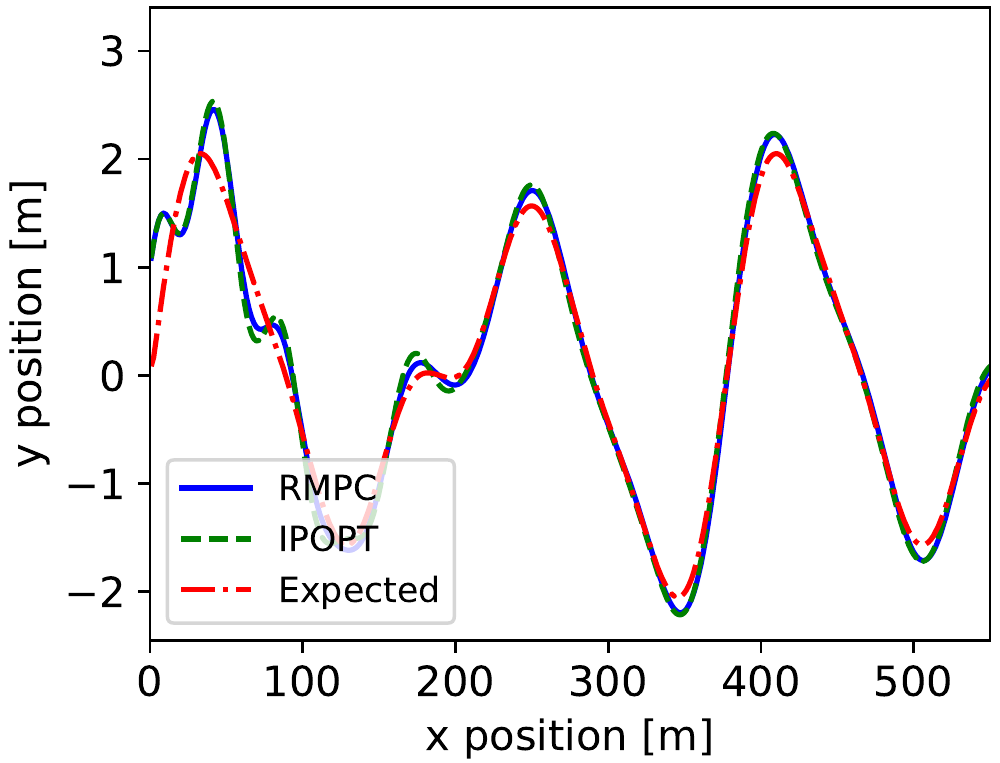}} \quad
\subfloat[\label{subFig:st2}]{\includegraphics[width=0.22\textwidth]{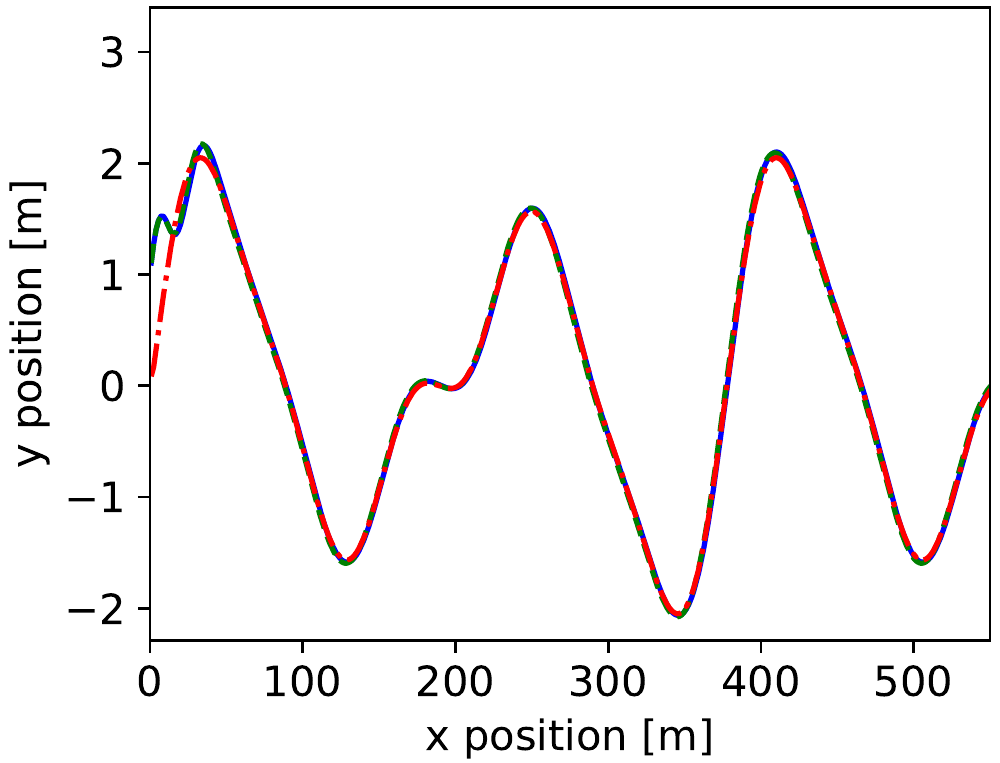}}\quad 
\subfloat[\label{subFig:st3}]{\includegraphics[width=0.22\textwidth]{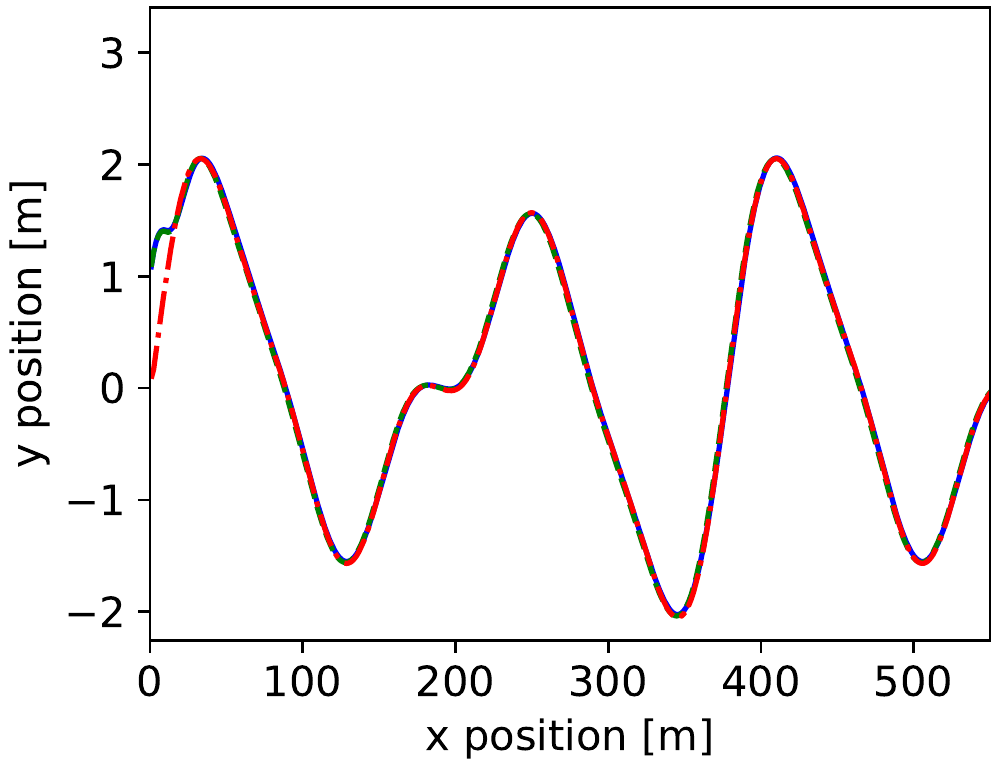}} \\
\subfloat[\label{subFig:sc1}]{\includegraphics[width=0.22\textwidth]{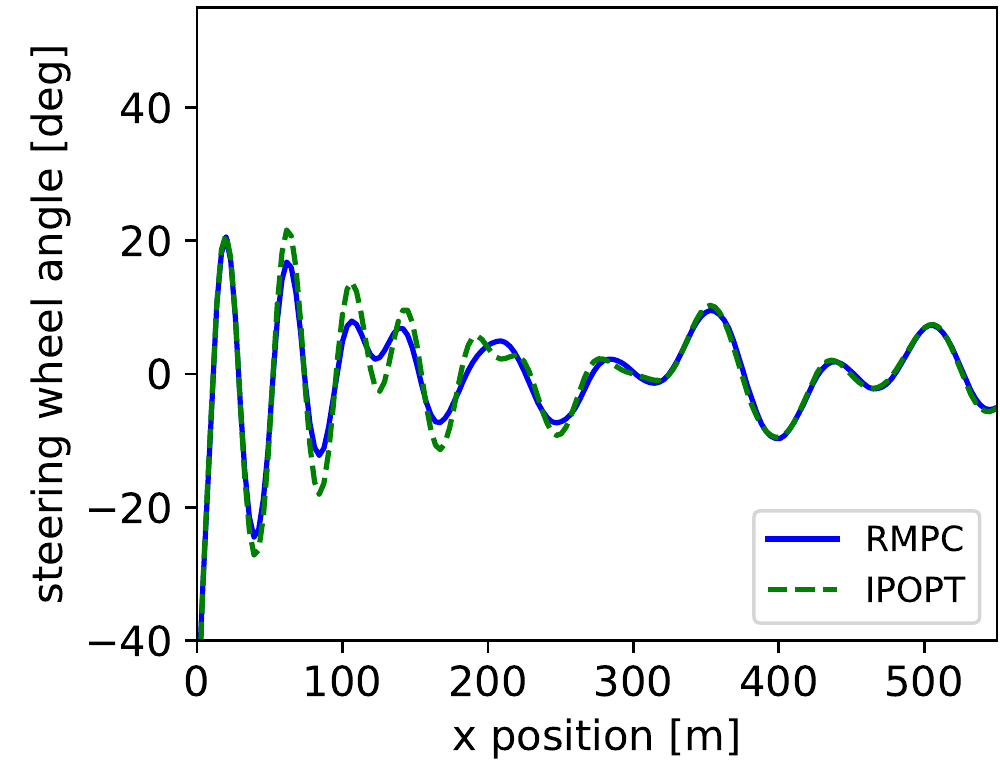}} \quad 
\subfloat[\label{subFig:sc2}]{\includegraphics[width=0.22\textwidth]{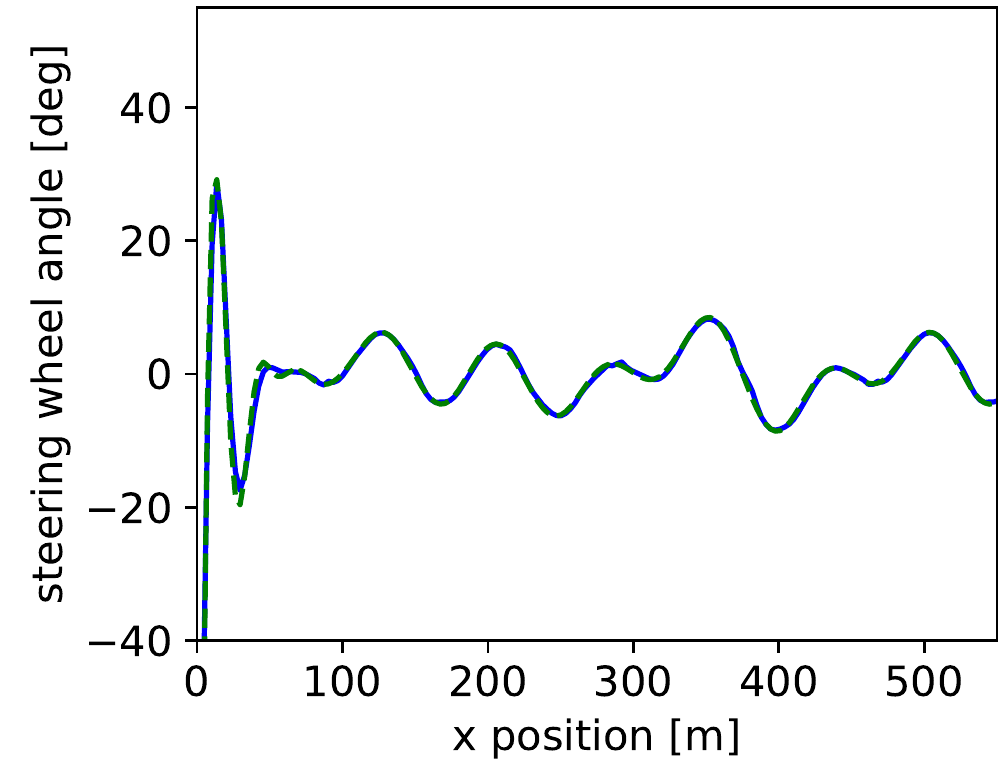}}\quad 
\subfloat[\label{subFig:sc3}]{\includegraphics[width=0.22\textwidth]{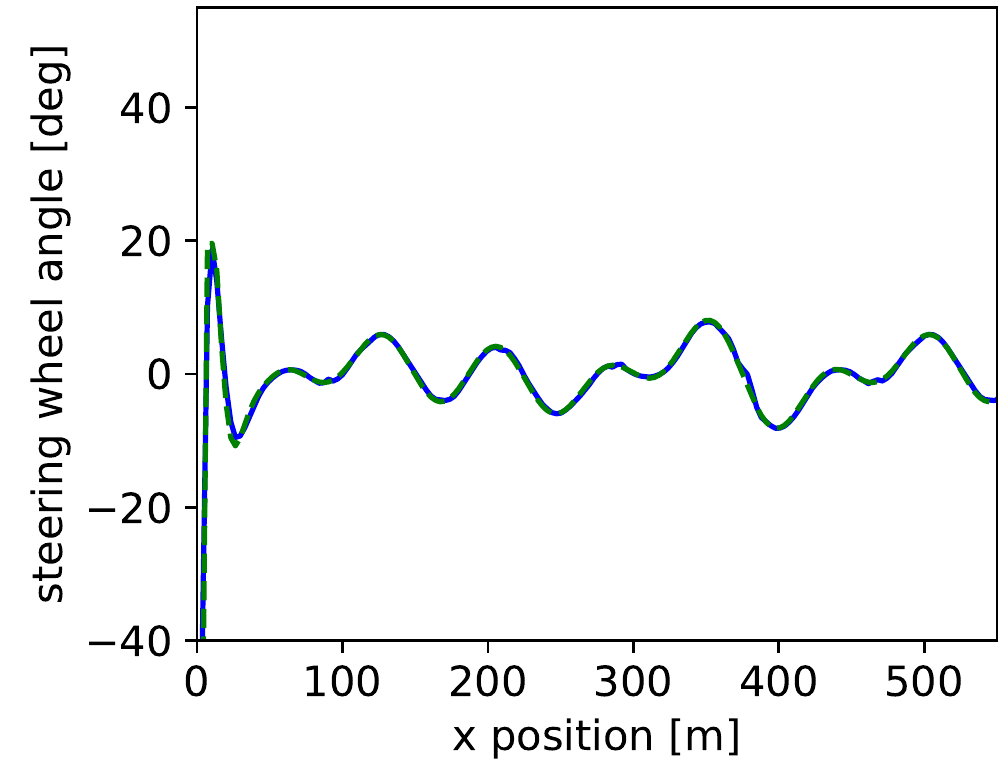}} 
\caption{Simulation results: trajectory and control inputs curves with different recurrent cycles $c$. (a) Trajectories when $c=7$. (b) Trajectories when  $c=11$.  (c) Trajectories when  $c=15$. (d) Control inputs when $c=7$. (e) Control inputs when $c=11$. (f) Control inputs when $c=15$.}
\label{f:comparison_linear}
\end{figure*}

In detail, Fig. \ref{f:comparison_linear} intuitively presents the trajectory curves and the corresponding control inputs of RMPC and the IPOPT solver. Obviously, the trajectory and control inputs generated by the RMPC controller almost overlap with the IPOPT controller. And the trajectory tracking error decreases significantly with the number of recurrent cycles. This explains the importance of adaptively selecting the optimal control input with the longest prediction horizon in real applications.

RMPC is an explicit MPC method, whose policy is learned offline based on predetermined system parameters. However, in practical applications, system parameters may vary due to online changes or inaccuracy measurements. To evaluate the robustness of RMPC  to different system parameters, we have changed certain system parameters and tested the performance of the fixed policy learned based on parameters in Table \ref{tab.parameters}.  The average tracking errors under different parameters are shown in Table \ref{tab.simulated_error}.  Results show that the final tracking performance is insensitive to the selected parameter changes. This implies that RMPC preserves good robustness property when the changes of system parameters are restricted within a reasonable range. This is mainly because the system usually behaves analogously as long as parameters do not change a lot. One can also add some parametric noises during the training process to further improve robustness. For those frequently changing and easily detected parameters, simply setting them as policy inputs can be a good choice.

\begin{table}[!htb]
\centering
\caption{Average Tracking Errors with Different System Parameters. The original parameters used for learning are shown in bold.}
	\label{tab.simulated_error}
		\begin{tabular}{l l l l l l l l}
		\hline\hline
$m$ [kg] &1200 &1300&1400&\textbf{1500}&1600&1700&1800\\
errors [cm]&1.76&1.65&1.54&1.46&1.39&	1.35&1.33
\\
\hline
$v_x$ [m/s] &13& 14 &15&\textbf{16} &17 &18 & 19 \\
errors [cm]&6.89&4.93&3.03&1.46&1.80&3.69&5.70\\
\hline
$\mu$ &0.7& 0.8 &0.9&\textbf{1.0} &1.1 &1.2 & 1.3 \\
errors [cm]&1.44&1.45&1.45&1.46&1.46&1.46&1.47\\
\hline
$k_1(\times 10^3)$ &-58 &-68 &-78 &\textbf{-88} &-98 &-108 &-118 \\
errors [cm]&1.61&1.40&1.39&1.46&1.56&1.67&1.77\\
\hline
$k_2(\times 10^3)$  &-64 &-74 &-84 &\textbf{-94} &-104 &-114  &-124 \\
errors [cm]&1.27&1.35&1.41&1.46&1.50&1.53&1.56\\
\hline\hline
\end{tabular}
\end{table}

To summarize, this example demonstrates the optimality, efficiency, and generality of the RMPC algorithm.

\section{Experimental Verification and Future Work}
\label{sec:experiment}
\subsection{Experimental Verification}
As shown in Fig. \ref{fig_HIL}, an IDRIVERPLUS four-wheeled robot is utilized to demonstrate the effectiveness of the proposed method in practical applications. For ease of understanding, this experiment is carried out by replacing the Carsim simulator in Fig. \ref{fig_HIL} with a real robot. Except for the vehicle parameters, all training details are the same as those in Section \ref{sec:simulation}. Note that the actual vehicle parameters are usually quite different from the theoretical model used for learning, since some parameters, such as tire cornering stiffness, are difficult to measure accurately. Therefore, the experimental results can also reflect the robustness of RMPC to inaccurate vehicle parameters.

We deployed the learned policy of RMPC in the four-wheeled robot, aiming to follow a sine-shaped reference path. Fig. \ref{f:comparison_experiment} shows the control results of RMPC and the IPOPT solver with different prediction steps. Although the variance of the steering wheel angle is larger than that in simulation due to the existence of system noise, both methods have achieved relatively good tracking performance. Table \ref{tab.tracking_error} compares the average tracking errors of these two methods under different number of prediction steps. Results show that RMPC achieved a smaller average tracking error than IPOPT in all cases. In particular,  when $c=3$, $c=9$, and $c=15$, RMPC reduces the tracking error by 56.2\%, 10.1\%, and 4.0\%, respectively. It is also obvious that the trajectory tracking error decreases significantly as the number of recurrent cycles increases, which provides evidence for the advantage of adaptively selecting the maximum prediction horizon. This real-world experiment demonstrates the  efficacy of RMPC in practical applications.
\begin{figure*}[!htb]
\centering
\captionsetup{justification =raggedright,
              singlelinecheck = false,labelsep=period, font=small}
\captionsetup[subfigure]{justification=centering}
\subfloat[\label{subFig:et1}]{\includegraphics[width=0.22\textwidth]{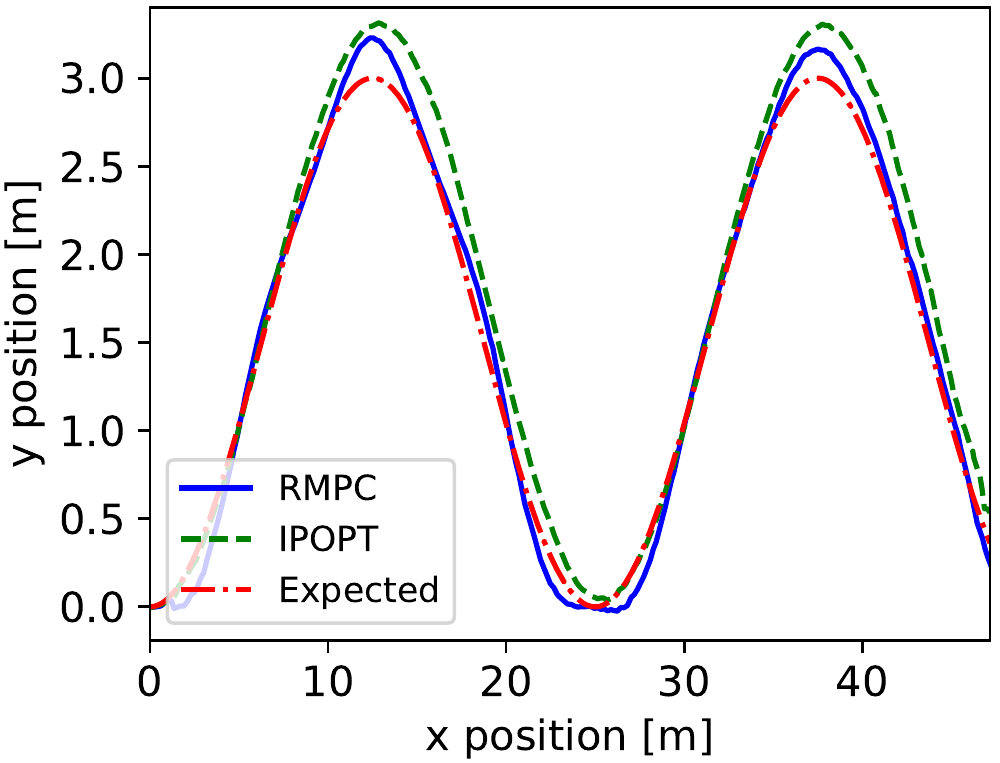}} \quad
\subfloat[\label{subFig:et2}]{\includegraphics[width=0.22\textwidth]{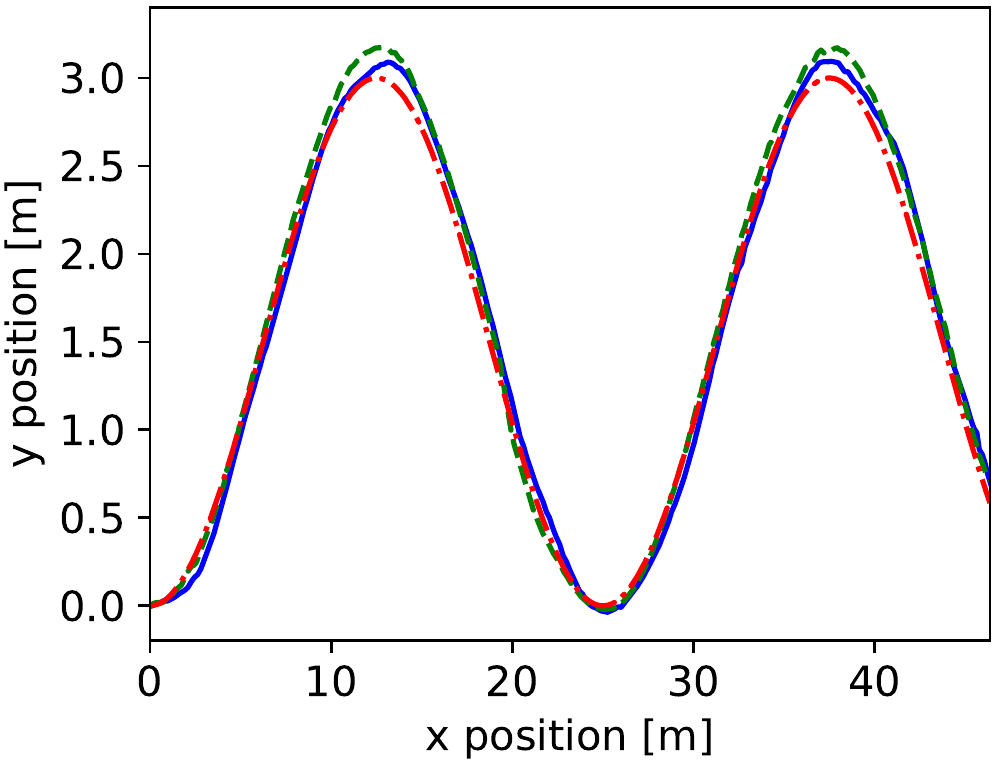}}\quad 
\subfloat[\label{subFig:et3}]{\includegraphics[width=0.22\textwidth]{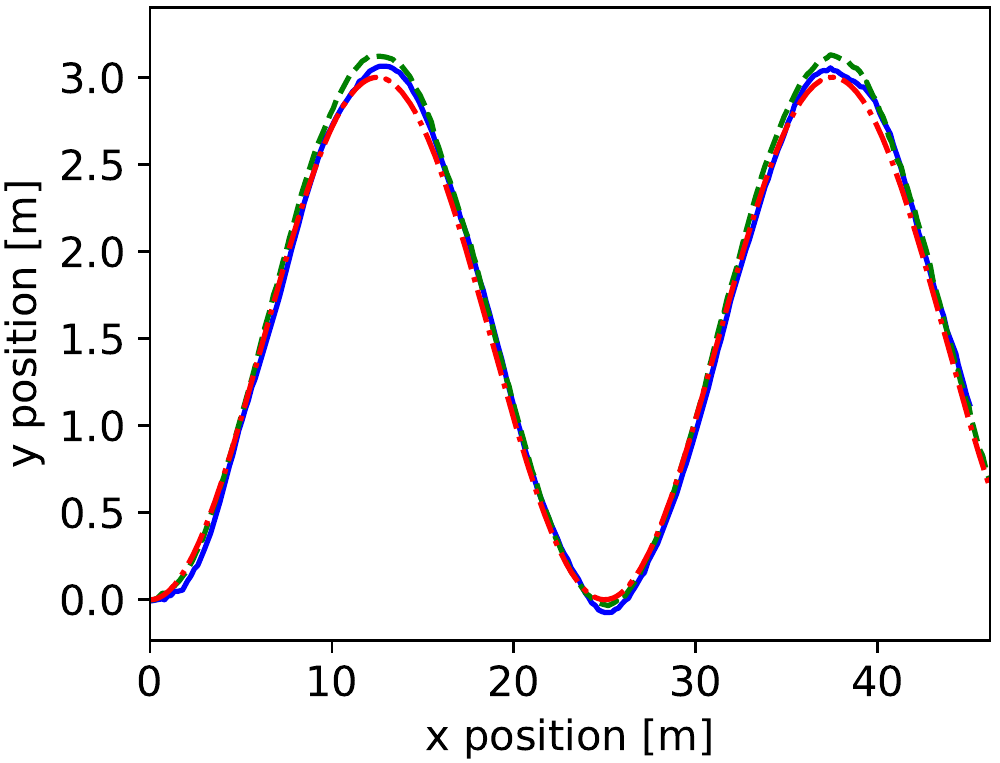}} \\
\subfloat[\label{subFig:ec1}]{\includegraphics[width=0.22\textwidth]{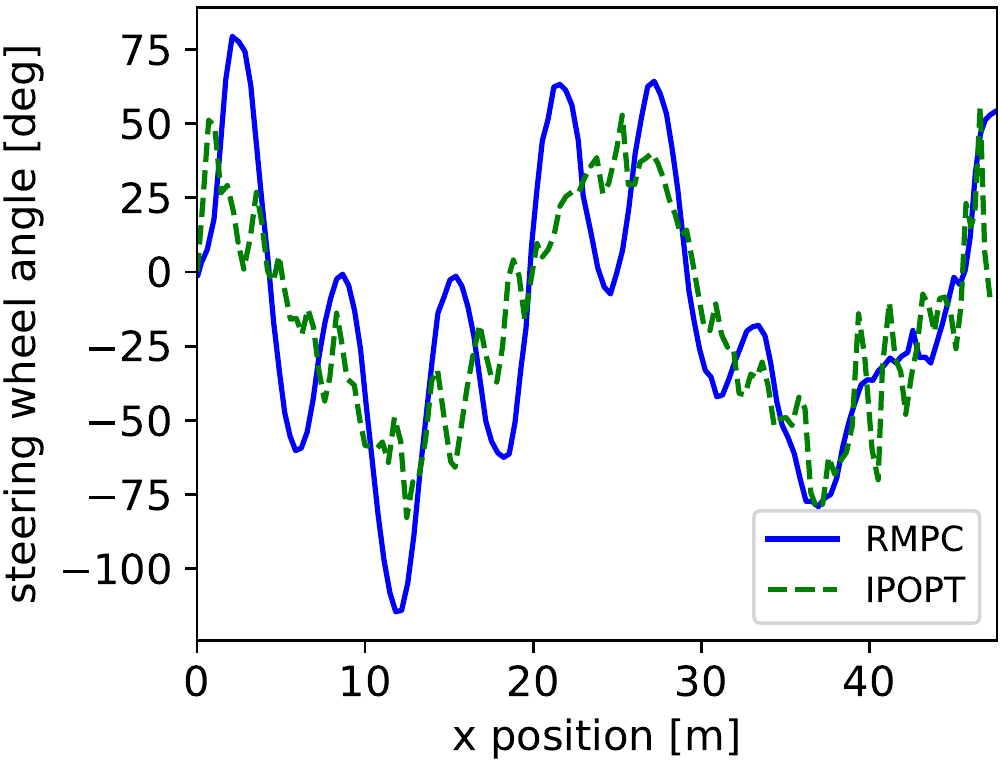}} \quad 
\subfloat[\label{subFig:ec2}]{\includegraphics[width=0.22\textwidth]{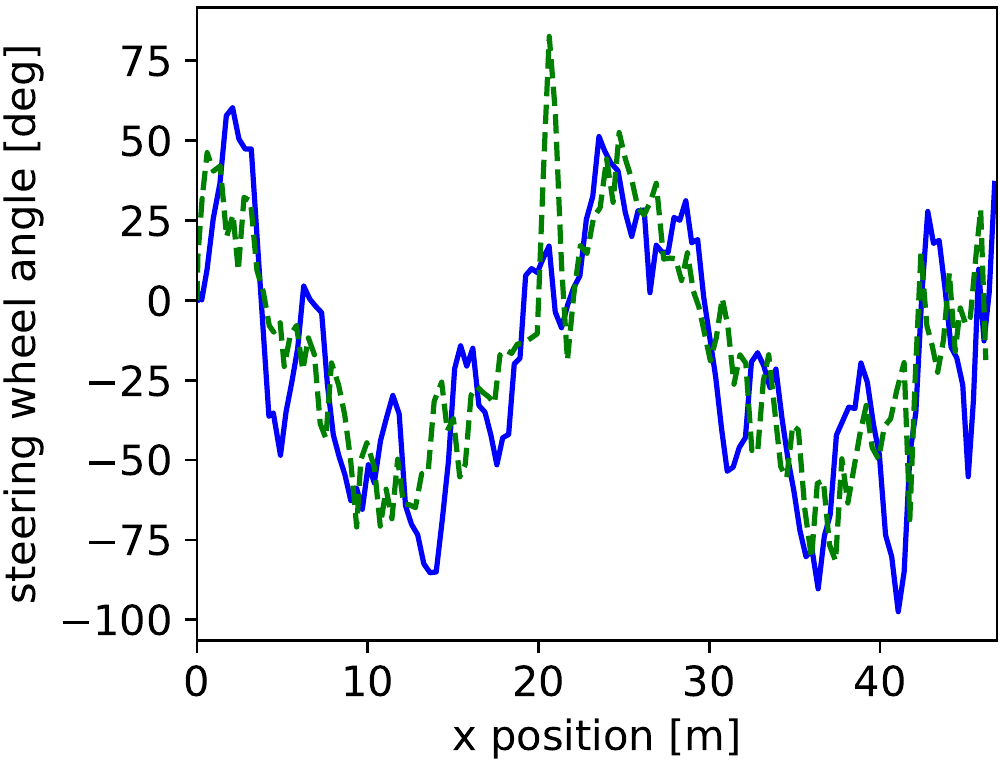}}\quad 
\subfloat[\label{subFig:ec3}]{\includegraphics[width=0.22\textwidth]{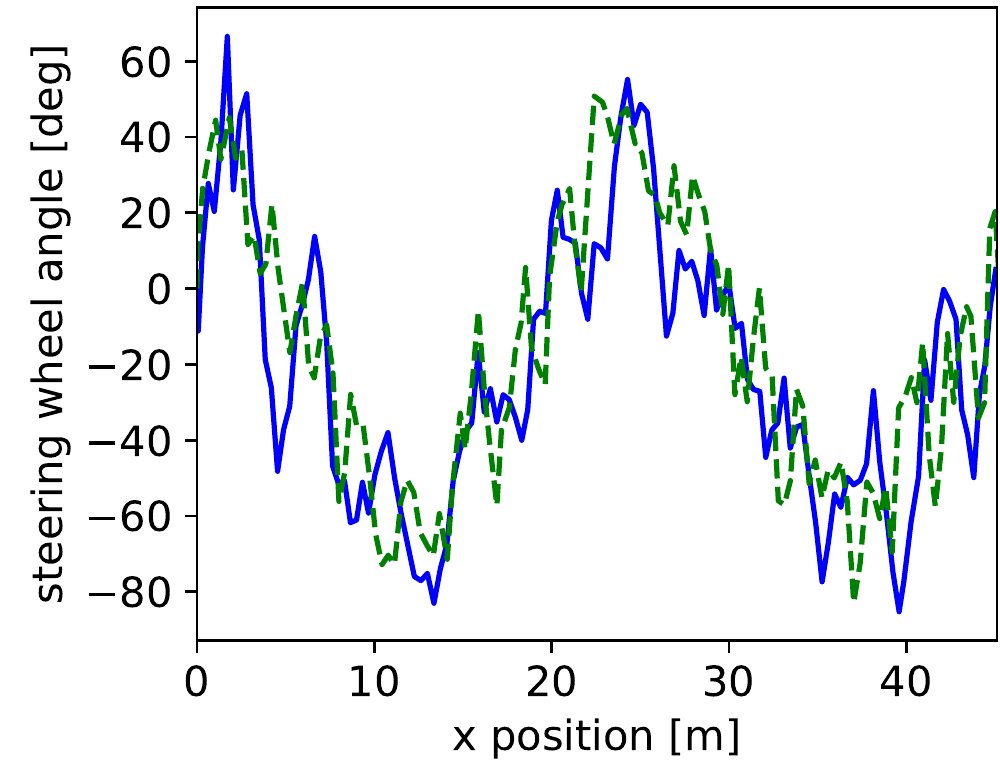}} 
\caption{Experiment results: trajectory and control inputs curves with different recurrent cycles $c$. (a) Trajectories when $c=3$. (b)Trajectories when  $c=9$.  (c) Trajectories when  $c=15$. (d) Control inputs when $c=3$. (e) Control inputs when $c=9$. (f) Control inputs when $c=15$.}
\label{f:comparison_experiment}
\end{figure*}

\begin{table}[!htb]
\caption{Average Tracking Errors}
\centering
\label{tab.tracking_error}
\begin{tabular}{l l l}
\hline\hline
Prediction step &RMPC&	IPOPT\\\hline
$c=3$ &10.43cm & 23.79cm \\
$c=9$ &9.74cm & 10.83cm \\
$c=15$ &6.75cm & 7.03cm \\
\hline\hline
\end{tabular}
\end{table}

\subsection{Limitations and Future Work}

In this paper, the proposed RMPC method is only suitable for problems without state constraints. In the future, we will extend RMPC to constrained cases by combining constrained policy optimization techniques \cite{duan2021adaptive}. Besides, the performance of the proposed RMPC method is evaluated only by the path tracking task of four-wheeled vehicles. More subsequent experiments for different systems and tasks will be addressed in further studies.  The future work also includes improving its robustness and investigating the stability of RMPC. 

\section{Conclusion}
\label{sec:conclusion}
This paper proposes the Recurrent Model Predictive Control (RMPC) algorithm to solve general nonlinear finite-horizon optimal control problems.
Unlike traditional MPC algorithms, it can make full use of the current computing resources and adaptively select the longest model prediction horizon. Our algorithm employs an RNN to approximate the optimal policy, which maps the system states and reference values directly to the control inputs. The output of the learned policy network after $N$ recurrent cycles corresponds to the nearly optimal solution of $N$-step MPC. A policy optimization objective is designed by decomposing the MPC cost function according to the Bellman's principle of optimality. The optimal recurrent policy can be obtained by directly minimizing the designed objective function, which is applicable for general nonlinear and non input-affine systems. The convergence and optimality of RMPC are further proved. We demonstrate its optimality, generality and efficiency using a HIL experiment. Results show that RMPC is over 5 times faster than the traditional MPC solver. The control performance of the learned policy can be further improved as the number of recurrent cycles increases. To prove its practicality, RMPC has also been applied to a real-world robot path-tracking task, and RMPC has achieved better path tracking accuracy than IPOPT in different prediction horizons.

\section*{Acknowledgment}
The authors are grateful to the Editor-in-Chief, the Associate Editor, and anonymous reviewers for their valuable comments.

\bibliographystyle{./Bibliography/IEEEtranTIE}
\bibliography{./Bibliography/BIB_xx-TIE-xxxx}

\end{document}